\documentclass[12pt,a4paper]{amsart}

\usepackage{ifthen,xifthen}
\usepackage[utf8]{inputenc}
\usepackage[T1]{fontenc}
\usepackage[centering]{geometry}
\usepackage{floatrow}
\usepackage{microtype}
\usepackage[ngerman,main=english]{babel}\usepackage{csquotes}
\usepackage{amsmath,amssymb,amsthm,mathtools}

\usepackage[textsize=tiny,shadow]{todonotes}

\numberwithin{equation}{section}
\theoremstyle{plain}
\newtheorem{theorem}{Theorem}
\newtheorem{lemma}[theorem]{Lemma}
\newtheorem{corollary}[theorem]{Corollary}
\theoremstyle{definition}
\newtheorem{definition}[theorem]{Definition}
\newtheorem{proposition}[theorem]{Proposition}
\theoremstyle{remark}
\newtheorem{remark}[theorem]{Remark}
\usepackage[shortlabels]{enumitem}
\usepackage{tikz}
\usetikzlibrary{calc,external,math}
\usetikzlibrary{datavisualization,datavisualization.formats.functions}
\RequirePackage{xcolor}
\providecommand{\red}[1]{\textcolor{red}{#1}}

\definecolor{darkred}{rgb}{0.5,0,0}
\definecolor{darkgreen}{rgb}{0,0.5,0}
\definecolor{darkblue}{rgb}{0,0,0.5}
\usepackage
  [ breaklinks          
  , colorlinks=true     
  , pdfborder={0 0 0}   
  , linkcolor=darkblue  
  , filecolor=darkgreen 
  , urlcolor=darkred    
  , citecolor=darkblue  
  ]{hyperref}

\usepackage[ style=alphabetic 
           , giveninits       
           , language=english
           , backref          
           , backend=biber
           ]{biblatex}
\renewbibmacro{in:}{
  \ifentrytype{article}
    {}
    {\printtext{\bibstring{in}\intitlepunct}}}
\AtEveryBibitem{\clearfield{month}}
\AtEveryBibitem{\clearfield{day}}

\DeclareFieldFormat{eprint:euclid}{Project Euclid:
  \ifhyperref
    {\href{http://projecteuclid.org/euclid.#1}{#1}}
    {http://projecteuclid.org/euclid.#1}%
}
\DeclareFieldFormat{eprint:hal}{%
  \ifhyperref
    {\href{https://hal.archives-ouvertes.fr/hal-#1}{hal-#1}}
    {https://hal.archives-ouvertes.fr/hal-#1}%
}
\DeclareFieldFormat{eprint:bioRxiv}{bioR$\chi$iv: %
  \ifhyperref
    {\href{https://www.biorxiv.org/content/#1}{#1}}
    {https://www.biorxiv.org/content/#1}%
}
\DeclareFieldFormat{eprint:medRxiv}{medR$\chi$iv: %
  \ifhyperref
    {\href{https://www.medrxiv.org/content/#1}{#1}}
    {https://www.medrxiv.org/content/#1}%
}
\DeclareFieldFormat{isbn}{
  \mkbibacro{ISBN}\addcolon\space
  \ifhyperref
    {\href{https://en.wikipedia.org/w/index.php?title=Special\%3ABookSources\&isbn=#1}{\nolinkurl{#1}}}
    {\nolinkurl{#1}}}

\usepackage[pdftex]{lscape}

\newcommand{\F}{\mathbb F}

\newcommand{\Z}{\mathbb Z}
\newcommand{\cF}{\mathcal F}
\DeclarePairedDelimiter\setsize\lvert\rvert
\DeclarePairedDelimiter\abs\lvert\rvert

\DeclarePairedDelimiterX\Icc[2][]{#1,#2}
\DeclarePairedDelimiterX\Ioc[2](]{#1,#2}
\DeclarePairedDelimiterX\Ico[2][){#1,#2}
\DeclarePairedDelimiterX\Ioo[2](){#1,#2}
\newcommand\pfn{p_{\mathrm{fn}}}
\newcommand\pfp{p_{\mathrm{fp}}}
\newcommand\Tfn{T_{\mathrm{fn}}}
\newcommand\Tfp{T_{\mathrm{fp}}}
\DeclareMathOperator\Prob{\mathbb P}
\DeclareMathOperator\E{\mathbb E}
\DeclareMathOperator\Var{Var}

\newcommand\given\mid
\newcommand\ifu[1]{\mathbf 1\ifthenelse{\equal{#1}_}{_}{_{#1}}}
\newcommand\kronecker[2]{\delta\ifthenelse{\equal{#1}_}{_{#2}}{_{#1,#2}}}
\newcommand\Union\bigcup
\newcommand\union\cup
\newcommand\Isect\bigcap
\newcommand\isect\cap

\newcommand\textq[1]{\text{#1}\quad}
\newcommand\qtextq{\quad\textq}

\newcommand\Id[1]{\mathrm{Id}\ifthenelse{\equal{#1}_}{_}{_{#1}}}
\DeclarePairedDelimiterX\spr[2]\langle\rangle{#1,#2}

\newcommand\ie{i.\,e.}

\newcommand\eg{e.\,g.}
\newcommand\cf{cf.}

\newcommand\helper{\xi}
\newcommand\COMP{\textrm{COMP}}
\newcommand\NCOMP{\textrm{NCOMP}}
\newcommand\nc{\delta}
\newcommand\sens{\mathrm{sens}}
\newcommand\spec{\mathrm{spec}}
\newcommand\typeI{\Prob( X_j = 0 \mid Z_j = 1 )}
\newcommand\typeII{\Prob( X_j = 1 \mid Z_j = 0 )}
\newcommand\rhodisj{\rho_{\textrm{disj}}}
\newcommand\rhoinfo{\rho_{\textrm{info}}}

\title[Statistics of noisy, one-stage group testing]
  {The statistics of noisy one-stage group testing in outbreaks}

\subjclass[2010]{05B05, 62J15, 92D30, 94A12}

\keywords{Group testing, linear regime, sensitivity, specificity, noise channel, pooled testing, non-adaptive testing, PCR}

\author[C.~Schumacher]{Christoph Schumacher}
\author[M.~T\"aufer]{Matthias T\"aufer}

\address{Christoph Schumacher, Fakult\"at für Mathematik, Technische Universit\"at Dortmund, 44227 Dortmund, Germany}
\email{christoph.schumacher@math.tu-dortmund.de}

\address{Matthias T\"aufer, Lehrgebiet Analysis, Fakult\"at Mathematik und Informatik, Fernuniversit\"at in Hagen, 58084 Hagen, Germany}
\email{matthias.taeufer@fernuni-hagen.de}

\date{\today}

\tikzmath{
  function gammakofrho(\x,\k,\q,\fp,\fn) { 
    return (1-\fp)*(1-(1-\fn)*\x)^(\q-\k);
  };
  function gammaoneofrho(\x,\q,\fp,\fn) { 
    return (1-\fp)*(1-(1-\fn)*\x)^(\q-1);
  };
  function binom(\n,\k) { 
    \result = 1.0;
    \km = min(\k, \n-\k);
    if \km>0 then {
      for \x in {1,...,\km}{
        \result = \result * (\n + 1 - \x) / \x;
      };
    };
    return int(round(\result));
  };
  function Bin( \n 
              , \p 
              , \k 
              ) {
    return binom(\n,\k) * \p^\k * (1-\p)^(\n-\k);
  };
  function BinRange(\n,\p,\m,\M) { 
    \result = 0;
    for \k in {\m,...,\M}{
      \result = \result + Bin(\n,\p,\k);
    };
    return \result;
  };
  function sensofrho( \x   
                    , \m   
                    , \afn 
                    , \q   
                    , \fp  
                    , \fn  
                    ) {
    \peff = 1 - \fp * gammaoneofrho(\x,\q,\fp,\fn);
    return BinRange(\m,\peff,\m-\afn,\m);
  };
  function csensofrho(\x,\m,\afn,\q,\fp,\fn) {
    return 1 - sensofrho(\x,\m,\afn,\q,\fp,\fn);
  };
  function cspecofrho(\x,\m,\afn,\q,\fp,\fn) {
    \peff = 1 - gammaoneofrho(\x,\q,\fp,\fn);
    return BinRange(\m,\peff,\m-\afn,\m);
  };
  function specofrho(\x,\m,\afn,\q,\fp,\fn) {
    return 1 - cspecofrho(\x,\m,\afn,\q,\fp,\fn);
  };
  function typeIofrho(\x,\m,\afn,\q,\fp,\fn) {
    \p0 = (1-\x) * cspecofrho(\x,\m,\afn,\q,\fp,\fn);
    \p1 = \x * sensofrho(\x,\m,\afn,\q,\fp,\fn);
    if equal(\p1,0) then { \result = 1; }
    else {
      if \p0 == 0 then { \result = 0; }
                  else { \result = \p0 / (\p0 + \p1); };
    };
    return \result;
  };
  function typeIIofrho(\x,\m,\afn,\q,\fp,\fn) {
    \p0 = (1-\x) * specofrho(\x,\m,\afn,\q,\fp,\fn);
    \p1 = \x * csensofrho(\x,\m,\afn,\q,\fp,\fn);
    if \p1 == 0 then { \result = 0; }
                else { \result = \p1 / (\p0 + \p1); };
    return \result;
  };
  function PZzeroofrho(\x,\m,\afn,\q,\fp,\fn) {
    return (1-\x) * specofrho(\x,\m,\afn,\q,\fp,\fn)
      + \x * csensofrho(\x,\m,\afn,\q,\fp,\fn);
  };
  function ET(\x,\m,\afn,\q,\fp,\fn,\n) {
    \a = \x * sensofrho(\x,\m,\afn,\q,\fp,\fn);
    \b = (1-\x) * cspecofrho(\x,\m,\afn,\q,\fp,\fn);
    return \n * (\a + \b);
  };
  function ETfp(\x,\m,\afn,\q,\fp,\fn,\n) {
    return \n * (1-\x) * cspecofrho(\x,\m,\afn,\q,\fp,\fn);
  };
  function ETfn(\x,\m,\afn,\q,\fp,\fn,\n) {
    return \n * \x * csensofrho(\x,\m,\afn,\q,\fp,\fn);
  };
  function VarTupper(\x,\m,\q,\n) {
    \h = 1 - (1-\x)^(\q-1);
    return \n*\m*\q*\x * (1-\x) * (1-\h^\m+\m*(\q-1)*(1-\x)^(\q-1)*\h^(\m-1));
  };
  function VarTfpupper(\x,\m,\q,\n) {
    \h = 1 - (1-\x)^(\q-1);
    return \n*\m*\q*\x * (1-\x) * (\h^\m+\m*(\q-1)*(1-\x)^(\q-1)*\h^(\m-1));
  };
  function polydegree(\k,\p) { 
    integer \degree;
    \degree = 0; 
    if \k>=1 then { \degree = ln(\k) / ln(\p); };
    return \degree;
  };
  function polysum(\a,\b,\p) { 
    integer \d, \k, \A, \B, \r;
    \d = polydegree(max(\a,\b),\p);
    \A = \a;
    for \k in {0,...,\d}{
      \x{\k} = int(Mod(\A,\p));
      \A = div(\A,\p);
    };
    \B = \b;
    for \k in {0,...,\d}{
      \y{\k} = int(Mod(\B,\p));
      \B = div(\B,\p);
    };
    for \k in {0,...,\d}{
      \z{\k} = int(Mod(\x{\k}+\y{\k},\p));
    };
    if \d<2 then {
      for \k in {\d+1,...,2}{
        \z{\k} = 0;
      };
      \d = 2;
    };
    \r = \z{\d};
    for \k in {\d-1,\d-2,...,0}{
      \r = \p*\r+\z{\k};
    };
    return \r;
  };
  function polyprod(\a,\b,\p) { 
    integer \da, \db, \k, \A, \B, \d, \r, \j, \yindex;
    \da = polydegree(\a,\p);
    \A = \a;
    for \k in {0,...,\da}{
      \x{\k} = int(Mod(\A,\p));
      \A = div(\A,\p);
    };
    \B = \b;
    \db = polydegree(\b,\p);
    for \k in {0,...,\db}{
      \y{\k} = int(Mod(\B,\p));
      \B = div(\B,\p);
    };
    \d = \da + \db;
    for \k in {0,...,\d}{
      \z{\k} = 0;
      for \j in {{max(0,\k-\db)},...,{min(\da,\k)}}{
        \yindex = \k-\j;
        \z{\k} = int(Mod(\z{\k} + \x{\j}*\y{\yindex},\p));
      };
    };
    if \d<2 then {
      for \k in {\d+1,...,2}{
        \z{\k} = 0;
      };
      \d = 2;
    };
    \r = \z{\d};
    for \k in {\d-1,\d-2,...,0}{
      \r = \p*\r+\z{\k};
    };
    return \r;
  };
  function Galoisprod(\a,\b,\p,\dz) {
    integer \c, \dx, \k, \C, \j, \ce, \deltax, \r;
    if equal(\dz,1) then { 
      \z{0} = 1;
      \z{1} = 1;
    } else {
      if equal(\p,2) then {
        if equal(\dz,2) then { 
          \z{0} = 1;
          \z{1} = 1;
          \z{2} = 1; 
        } else {
          if equal(\dz,3) then { 
            \z{0} = 1;
            \z{1} = 1;
            \z{2} = 0;
            \z{3} = 1; 
          } else {
            if equal(\dz,4) then { 
              \z{0} = 1;
              \z{1} = 1;
              \z{2} = 0;
              \z{3} = 0;
              \z{4} = 1; 
            } else {
              if equal(\dz,5) then { 
                \z{0} = 1;
                \z{1} = 0;
                \z{2} = 1;
                \z{3} = 0;
                \z{4} = 0;
                \z{5} = 1; 
              } else {
                \errmessage{Please provide irreducible polynomial for \p^\dz}
              };
            };
          };
        };
      } else {
        if equal(\p,3) then {
          if equal(\dz,2) then { 
            \z{0} = 2;
            \z{1} = 2;
            \z{2} = 1; 
          } else {
            if equal(\dz,3) then { 
              \z{0} = 1;
              \z{1} = 2;
              \z{2} = 0;
              \z{3} = 1; 
            } else {
              \errmessage{Please provide irreducible polynomial for \p^\dz}
            };
          };
        } else {
          if equal(\p,5) then {
            if equal(\dz,2) then { 
              \z{0} = 2;
              \z{1} = 4;
              \z{2} = 1; 
            } else {
              \errmessage{Please provide irreducible polynomial for \p^\dz}
            };
          } else {
            \errmessage{Please provide irreducible polynomial for \p^\dz}
          };
        };
      };
    };
    \c = polyprod(\a,\b,\p);
    if \c < \p^\dz then {
      \r = \c;
    } else {
      \dx = polydegree(\c,\p);
      \C = \c;
      for \k in {0,...,\dx}{
        \x{\k} = int(Mod(\C,\p));
        \C = div(\C,\p);
      };
      if \dx == \dz then { \dx = \dx + 1; \x{\dx} = 0; };
      for \k in {\dx,\dx-1,...,\dz}{
        \ce = \k - \dz; 
        for \j in {0,...,\dz}{
          \deltax = \ce + \j;
          \x{\deltax} = int(Mod(\x{\deltax}-\x{\k}*\z{\j},\p));
        };
      };
      if \dz<2 then {
        for \k in {\dz+1,...,2}{
          \x{\k} = 0;
        };
        \dz = 2;
      };
      \r = \x{\dz};
      for \k in {\dz-1,\dz-2,...,0}{
        \r = \p*\r+\x{\k};
      };
    };
    return \r;
  };
}%
\def\colorlist{{"red","blue!60","green","orange!70","yellow","pink","cyan","violet!70","gray"}}%

\pgfkeys{/pgf/number format/.cd,fixed,precision=2}%
\def\twodecimals#1{%
  \pgfmathprintnumber{#1}%
}

\begin{document}
\maketitle

\begin{abstract}
  In one-stage or non-adaptive group testing,
  instead of testing every sample unit individually,
  they are split, bundled in pools, and simultaneously tested.
  The results are then decoded to infer the states of the individual items.
  This combines advantages of adaptive pooled testing,
  \ie\ saving resources and higher throughput,
  with those of individual testing,
  \eg\ short detection time and lean laboratory organisation,
  and might be suitable for screening during outbreaks.
  \par
  We study the COMP and NCOMP decoding algorithms
  for non-adaptive pooling strategies based on
  maximally disjunct pooling matrices with constant row and column sums
  in the linear prevalence regime and
  in the presence of noisy measurements motivated by PCR tests.
  We calculate sensitivity, specificity,
  the probabilities of Type I and II errors,
  and the expected number of items with a positive result
  as well as the expected number of false positives and false negatives.
  We further provide estimates on the variance
  of the number of positive and false positive results.
  \par
  We conduct a thorough discussion of the calculations and bounds derived.
  Altogether, the article provides blueprints for screening strategies
  and tools to help decision makers to appropriately tune them in an outbreak.
\end{abstract}

\section{Introduction}

Group testing addresses the problem of detecting a rare feature
in a large population.
By pooling sample units,
one can often clear large subsets of the sample with a single test.
After this first stage,
classical group testing proceeds to retest items in positive pools.
Thus, one splits the sample units into several pieces beforehand,
pools only part of each item and keeps the rest for retesting
in subsequent stages.

While this can saves resources,
these adaptive testing strategies are time-consuming and laborious
which is one reason preventing their wide-spread implementation
\cite{RKI-20}.
A way to mitigate these drawbacks is \emph{one-stage}
or \emph{non-adaptive} group testing:
In order to avoid a second stage,
one includes (parts of) each sample unit in several pools
and exonerates every sample unit which appears in a pool that tests negative.
This decoding strategy is called
\emph{Combinatorial Orthogonal Matching Pursuit} (COMP).

The tradeoff is a non-zero probability for false positive test results
which occur if a negative sample unit is ``shadowed'' by positive items,
that is, if each pool that contains the falsely positive item is contaminated
by actually positive items.
To minimize shadowing, we commit to particular pool designs referred to as
\emph{multipools} in~\cite{Taeufer-20}, 
which are based on \emph{maximally disjunct} pooling matrices.
This means that each pair of sample units meets in at most one pool.

Still, the probability for falsely positive results has to be controlled
for a reliable interpretation of the test results.
To this end, we assume that every sample unit is independently infected
with probability~$\rho$.
This scenario is often referred to the \emph{linear regime}
and is a natural assumption in population screening.
We also account for measurement errors with a noise model
inspired by biomedical testing which has been argued for in \cite{BaronRZ-20}
and is given in formula~\eqref{eq:noise-model}.
This can introduce false negative results
which we counter by the error-correcting noisy COMP (NCOMP) algorithm.

In this setting, we provide formulas for sensitivity, specificity,
the probabilities of Type~I and II errors, and the expected number
of positive, false positive and false negative results.
We also provide bounds on the variance of the number of
positive and false positive results in the noiseless case.

\subsection{Motivation: screening via PCR}

Real time reverse transcription polymerase chain reaction
(RT-PCR or briefly PCR)
is a biochemical procedure to identify certain DNA or RNA sequences
and an important tool to detect infectious diseases.
Its large scale use can be constrained by factors such as
the availability of collection devices (swab kits),
trained staff to take samples, and their protective equipment,
the availability of reagents, the number of PCR machines,
lab staff, and logistics.

In epidemiological scenarios,
there are different regimes of PCR application to distinguish.
\emph{Diagnostic testing} happens in the clinical context with the goal
to precisely measure the viral or bacterial load in a patient
and inform clinical treatment.
One wants to \emph{maximize accuracy} and \emph{minimize detection time}
whereas an efficient use of resources or costs are of secondary concern.
In contrast to that, \emph{screening} takes place in a public health context,
and the goal is to
\emph{maximize the overall epidemiological or public health benefit}
with given resources.
This typically means that one wants to prevent as many transmission events
as possible --- usually by identifying and isolating infectious carriers
who might be pre-symptomatic or asymptomatic.

A screening strategy increasing the overall number of people tested
could therefore be justified --- even if it leads to a reduced accuracy
of single tests since this could be compensated by frequent retesting.
This might be achieved by a range of measures such as self-swabbing,
testing saliva instead of nasopharyngeal swabs~\cite{Vogels2020},
running the PCR for fewer cycles
(testing for infectiousness instead of infection),
optimization of the use of critical reagents in the lab,
and also by adaptive and non-adaptive pooling, which we focus on here.

In the COVID-19 pandemic, large-scale screening has been suggested
as an effective measure~\cite{ECDC-20}
and pooled testing has been suggested as an approach to deal
with scarce resources~\cite{Mallapaty2020,Taipale2020}.
There has been some emphasis on \emph{two-stage, adaptive testing},
where, after a first round of pooled tests,
individuals in positive pools are assessed again~\cite{Liu2020,Chen2020}.

One-stage strategies, where results become available
after only one round of testing, have been suggested in
\cite{Shental2020,Taeufer-20,Ghosh2020,BaronRZ-20,PetersenBJ-20}.
It has been shown that detection of the SARS-CoV-2 virus
in pools of size 100 is possible,
which promises massive improvements in throughput~\cite{Mutesa2020}. 
However, the implementation of pooling strategies will require
a thorough understanding of the consequences such as possible tradeoffs
in accuracy involved.
In this paper, we aim to contribute towards such an understanding
by investigating statistical measures
associated with one-stage pooling strategies.

\subsection{Non-adaptive group testing}

We focus on \emph{non-adaptive} or \emph{one-stage group testing}
where every person's sample is put into a number of pools according
to a design matrix, all pools are tested in parallel,
and the results are then decoded.
We consider this as preferable to two- or multi-stage strategies
since only one round of the PCR is required which offers shorter detection times
and possibly a leaner organization of laboratory processes~\cite{Taeufer-20}.
In the context of the COVID-19 pandemic this is particularly important
because the viral load of the SARS-CoV-2 virus and the infectiousness
have been observed to be high in patients before and around symptom onset
\cite{He2020,To2020,Adam2020,Kupferschmidt2020}.
Thus, every hour between the sample taken and the result returned matters.
We consider only binary PCR,
where results are ``positive'' or ``not positive''.
There exists approaches on pool testing for COVID-19 where, using compressed sensing, also quantitative results of the PCR are used in the reconstruction~\cite{Ghosh2020,PetersenBJ-20}. 
If a patient is identified in a screening process,
this should probably be followed up by an individual test for clinical purposes,
but this would belong to the realm of diagnostics
which is not the topic of this note.

Our pooling strategies will be based on design matrices
with constant row and column sum and which have maximal disjunctness.
Such designs have also been studied in combinatorics,
where they are known as Steiner systems \cite{ColbournD-07}
and have been called \emph{multipools} in \cite{Taeufer-20}.
Constant row and column designs have been seen to be practical
\cite{ErlichEtAl2015}, and disjunctness is directly related
to the maximal number of infected items for which perfect reconstruction
is mathematically guaranteed in a noiseless scenario, see \cite{AldridgeJS-19}.

We construct examples of such design matrices using
linear Reed-Solomon Codes~\cite{ReedS-60, KautzS-64}
and the Shifted Transversal Design~\cite{Thierry-Mieg2006},
two constructions based on the same underlying algebraic principle.
We will consider noisy measurements where the noise model
\begin{equation}\label{eq:noise-model}
  \Prob\left(
    \text{Pool tests negative}
    \given
    \text{Pool contains~$k$ positive items}
  \right)
  =(1-\pfp)\pfn^k
\end{equation}
depends on the number of true infected items in a pool
and contains two parameters~$\pfp$ and~$\pfn$,
modulating false positive and false negative probabilities
of a single measurement, respectively.
This noise model is for instance argued for in~\cite{BaronRZ-20}.
We will investigate the ``simple'' COMP decoding algorithm,
cf.~\cite{AldridgeJS-19} for an overview, and its error-correcting brother,
the NCOMP decoding algorithm~\cite{Chan2011,ChunLamChan2014}.
Both are trivial to implement
with minimal run-time and storage.

\subsection{Statistical measures of non-adaptive group testing}

In order to evaluate a testing strategy, several quantities can be considered.
They may depend on~$\pfn$, $\pfp$, the prevalence $\rho\in\Ioo01$,
and parameters such as the pool size~$q$,
the number~$m$ of pools each item participates in,
and a parameter~$\nc$ tuning the NCOMP algorithm.

The first quantity is the \emph{compression ratio},
that is the inverse of the average number of tests required per item.
It describes the savings compared to individual testing.
Note that in two-stage strategies, not only the average number of tests,
but also the variance or standard deviation of the number of tests used
for a given population size are relevant.
This is because the number of tests needed in the second round is unknown
and this process of re-assessing pools can create logistical challenges.
In one-stage strategies,
the number of tests per item is a fixed number and has zero variance.
We consider this another advantage of non-adaptive versus adaptive testing.

The savings in tests are to be compared to possible sacrifices in accuracy.
In the literature one finds investigations of:
\begin{itemize}
\item The maximal number of infected items which are guaranteed
  to be correctly identified if there is no noise, \cite{Mazumdar2012},
\item the minimal number of tests required to achieve asymptotically
  a full reconstruction in the sublinear prevalence regime, \cite{AldridgeJS-19}.
\end{itemize}
Both metrics might not be ideal for the application to screening ---
on the one hand because they do not take into account noise,
on the other hand because they are either tailored towards worst-case scenarios
or work in asymptotic limits and in the so-called sublinear regime
where the portion of infected persons is assumed to tend to zero
with growing population size.

Instead, one would rather like to study the \emph{average performance}
or average number of false positive results~\cite{Mazumdar2012}.
In the literature, one finds investigations where for a
\emph{random draw of a fixed number of infected items}
quantities such as
\begin{itemize}
\item the \emph{number of positive results}~$T$,
\item the \emph{number of false positive results} $\Tfp$,
\item and the \emph{number of false negative results} $\Tfn$
\end{itemize}
are simulated~\cite{Shental2020}.
Such a fixed number of infected patients in a pool is a simplifications
which ignores the true probabilistic structure of the infections.
We believe that an approach better suited to inform decision making
is to investigate for a given prevalence~$\rho$.
\begin{itemize}
\item the \emph{sensitivity}, that is the probability that an infected item
  is actually picked up by the testing strategy:
  \begin{equation}\label{eq:sensitivity}
    \sens
    =\Prob(\text{test result positive}\given\text{patient infected})\text,
  \end{equation}
\item the \emph{specificity}, that is the probability that a non-infected item
  is correctly identified as negative:
  \begin{equation}
    \label{eq:specificity}
    \spec
    =\Prob(\text{test result negative}\given\text{patient not infected})\text,
  \end{equation}
\end{itemize}
In addition to that, there are two more quantities which matter
from a public health perspective,
since they tell an individual how reliable their result is.
Indeed, there exist situations where a testing strategy has both
high sensitivity and high specificity but still most positive results
will be false positives, see Table \ref{table:sensitivity_and_specificity}
for a synthetic example.
This phenomenon is also known as \emph{screening paradox}
and can be disadvantageous since patients might be reluctant
to comply with public health measures based on these probabilities.

\begin{table}[h!]
  \begin{tabular}{c|c|c}
    $n = 1000$ & patient infected & patient not infected \\
  \hline
    result positive & $19$ & $20$ \\
  \hline
    result negative & $1$ & $960$ \\
  \end{tabular}
\caption{This synthetic example of a test run on $n=1000$ individuals
  has both reasonably high observed empiric sensitivity and specificity
  ($19/20 = 0.95$ and $960/980 \approx 0.98$),
  but more than half ($20/39 \approx 0.51$)
  of positive results are actually false positives.}
  \label{table:sensitivity_and_specificity}
\end{table}

Therefore, we also quantify:
\begin{itemize}
\item the Type I error, that is the probability that a positive test results
  turns out to be a false positive:
  \begin{equation*}
    \Prob(\text{patient not infected}\given\text{test positive})\text,
  \end{equation*}
\item the Type II error, that is the probability that a negative test result
  is a false negative:
  \begin{equation*}
    \Prob(\text{patient infected}\given\text{test negative})\text.
  \end{equation*}
\end{itemize}

We also calculate the expected number of positive results,
false positive results, and false negative results.

Finally, we also estimate the variance of the number of positive
and false positive tests in the noiseless case.
We use the Efron--Stein inequality for this bound.
Such an estimate is important because the number of false positive results
seems to be a heavy tailed random variable:
Most of the time, most items will be correctly identified, but in some rare cases,
when the random number of infected items in the pools exceeds a certain threshold,
a phase transition occurs and an overwhelming number of items in the test
will be erroneously flagged as positive.
Some authors suggest to treat this phenomenon as a
``graceful failure''~\cite{Ghosh2020}
which might still flag a local outbreak
without specifying the infected individuals.
However, in order to correctly flag this phenomenon,
more knowledge on higher moments
of the number of positive and false positive results is useful.
Hence we provide this bound on the variance.

The rest of the article is structured as follows:
In Section~\ref{sec:notation}, the main definitions and notations are introduced.
Section~\ref{sec:results} contains the main results.
After that, Section~\ref{sec:calculations}
contains the calculations of sensitivity,
specificity and the probability of Type I and Type II errors.
In Section~\ref{sec:Efron-Stein}, the bounds of the variance are proved,
and Section~\ref{sec:matrices} provides details on the construction
of some non-adaptive pooling matrices of the form we consider.

\section{Notation and results}\label{sec:notation_and_results}

\subsection{Notation}\label{sec:notation}

We use notation inspired by the group testing literature:
There are~$n$ \emph{items} (\eg\ nasophrygnal swabs)
which can be \emph{infected} or \emph{non-infected}.
The \emph{state of the items} is a vector:
\begin{equation*}
  X=(X_j)_{j\in\{1,\dotsc,n\}} \in\{0,1\}^n\text,
\end{equation*}
where $X_j=1$ if item~$j$ is infected and $X_j=0$ otherwise.

We assume that the~$X_j$ are drawn independently from a Bernoulli distribution
with \emph{infection probability} or \emph{prevalence} $\rho\in\Ioo01$,
where for practical purposes,~$\rho$ is assumed small.
This is also called the \emph{linear prevalence regime}
in group testing.

The items are pooled into pools of size~$q$
such that every item participates in exactly~$m$
pools and such that no pair of items appears in more than one pool.
In particular, \emph{the overall number of tests} is $t=mq$
and the \emph{compression ratio},
the factor of improvement with respect to individual testing, is $n/t$.

Formally, the pooling can be described by the \emph{pooling matrix}
\begin{equation*}
  A\in\{0,1\}^{t\times n}
\end{equation*}
which encodes which item is put into which pool.
We write $A_{i,j}=1$ if and only if item~$j$ enters into pool~$i$.
In particular, $(AX)_i$ is the number of positive items in pool~$i$.
In terms of the pooling matrix~$A$,
the above conditions mean that~$A$ is a \emph{multipool matrix}
in the sense of the following definition.
\begin{definition}\label{def:multipool}
  We call the matrix $A\in\{0,1\}^{t\times n}$
  an \emph{$(n,q,m)$-multipool matrix} \cite{Taeufer-20}
  or a \emph{$(m-1)$-disjunct matrix with constant row and column sums},
  if the following three conditions hold:
  \begin{enumerate}[(M1)]
  \item The sum over every row is~$q$.
  \item The sum over every column is~$m$.
  \item The scalar product of any two columns is at most one.
  \end{enumerate}
\end{definition}
In the language of \emph{Block designs},
multipools are known as \emph{uniform $1$-designs} or \emph{Steiner systems}
and a maximal multipool with $n=q^2$ and $m=q+1$ is a $2$-design,
\cf\ \cite{ColbournD-07,Stinson-04}.

Definition~\ref{def:multipool} imposes constraints on the interplay of~$n$,
$q$, and~$m$.
$(n,q,m)$-multipool matrices exist for instance if~$q$ is a \emph{prime number}
or \emph{a power of a prime}, the overall number of items is $n=q^2$,
and~$m$ is not larger than $q+1$.
We will provide the details on this particular construction
in Theorem~\ref{thm:matrices} and Section~\ref{sec:matrices}.

A pool can test \emph{positive} or \emph{negative}.
The \emph{pool test results} are a vector
\begin{equation*}
  Y=(Y_i)_{i\in\{1,\dotsc,t\}}\in\{0,1\}^t\text,
\end{equation*}
where $Y_i=1$ if and only if pool~$i$ tests positive.
In particular, $(A^\top Y)_j$
is the number of pools containing~$j$ that tested positive.

The testing process is assumed to be noisy according to the noise model
\begin{equation}\label{eq:error_model}
  \Prob(Y_i=0\given(AX)_i=k)
  =(1-\pfp)\pfn^k\text.
\end{equation}
It depends on the number of positives in a pool
as well as on two parameters $\pfp,\pfn\in\Icc01$,
the false positive and false negative probability.
The error model~\eqref{eq:error_model} is for instance argued for in~\cite{ZhuRiveraBaron2020},
and we note that the false positive probability~$\pfp$
and false negative probability~$\pfn$
will in practice depend on the pool size~$q$, \ie\ the dilution.
Since we are reluctant to argue here for an error model
incorporating dilution due to pool size,
we treat~$\pfp$ and~$\pfn$ as parameters which will depend on~$q$
and need to be inferred from experiments.

The results of the pools are then decoded using the COMP
or NCOMP decoding algorithm.
\begin{definition}
  Let an $(n,q,m)$-multipool matrix be given.
  The \emph{Noisy Combinatorial Orthogonal Matching Pursuit}
  decoder with parameter~$\nc\in\{0,1,\dotsc,m\}$,
  abbreviated as $\NCOMP(\nc)$,
  declares an item as tested positive if and only if at most~$\nc$
  of the~$m$ pools which contain item~$j$ are not tested positive:
  \begin{equation*}
    \sum_{j\in\{1,\dotsc,t\}\colon A_{i,j}=1}Y_j
    =\sum_{i\in\{1,\dotsc,t\}}A_{i,j}Y_i
    =(A^\top Y)_j
    \ge m-\nc\text.
  \end{equation*}
  In the special case $\nc=0$, when an item is declared positive
  if and only if all of its pools test positive,
  this decoder is simply called the
  \emph{Combinatorial Orthogonal Matching Pursuit} decoder:
  $\COMP:=\NCOMP(0)$.
  If we do not want to specify the parameter~$\nc$,
  we simply write NCOMP.
\end{definition}

\begin{remark}
	COMP has been described by numerous authors where~\cite{KautzS-64}
  seems to be the first occurrence.
	The names COMP and NCOMP themselves seem to have been coined in~\cite{Chan2011}.
	We refer to~\cite{AldridgeJS-19} for a more thorough discussion.
	We also note that there exist other decoding algorithms
  in the mathematical literature such as the Definite Defective (DD)
  algorithm~\cite{AldridgeBJ-14}
  and the algorithm in~\cite{CojaOglanGHL-20}
  which relies on random constant column designs
  and which has been shown to be information-theoretically optimal
  in the~\emph{sublinear prevalence regime}.
	Furthermore, performance guarantees on COMP and DD in the sublinear regime
  have recently been investigated in~\cite{GebhardJLR-20}.
\end{remark}
	
The decoded results are a vector
\begin{equation*}
  Z=(Z_j)_{j\in\{1,\dotsc,n\}}\in\{0,1\}^n\text,
\end{equation*}
where $Z_j=1$ if COMP or NCOMP declares item~$j$ as positive and~$Z_j = 0$ otherwise.

\subsection{Results}\label{sec:results}

We first ensure the existence of pooling matrices
as in Definition~\ref{def:multipool}.

\begin{theorem}\label{thm:matrices}
  Let the pool size~$q$ be a prime number
  or a power of a prime and let the total number of items be $n=q^2$.
  Then, $(n,q,m)$-multipools exist if and only if the number of pools~$m$
  an item participates in satisfies $m\le q+1$.
\end{theorem}

Such matrices have been studied in the literature and similar structures have been suggested for pooling strategies.
If~$q$ is a prime number, such matrices can be constructed by
the Shifted Transversal Design~\cite{Thierry-Mieg2006}.
If~$q$ is a power of a prime, they can be constructed by Reed-Solomon codes
and have been suggested for pooling \eg\ in~\cite{ErlichEtAl2015}.
We provide details on the construction of such matrices and illustrations
in Section~\ref{sec:matrices}.
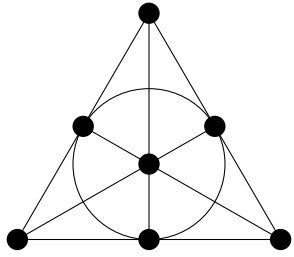
\begin{figure}[ht]
  \tikzsetnextfilename{fano}
  \begin{tikzpicture}
    \draw (210:2) -- (330:2) -- (90:2) -- cycle
          (210:2) -- (30:1)
          (330:2) -- (150:1)
          (90:2) -- (270:1);
    \draw (0,0) circle[radius=1cm];
    \fill[black] (0,0) circle[radius=4pt];
    \fill[black] (210:2) circle[radius=4pt];
    \fill[black] (330:2) circle[radius=4pt];
    \fill[black] (90:2) circle[radius=4pt];
    \fill[black] (30:1) circle[radius=4pt];
    \fill[black] (150:1) circle[radius=4pt];
    \fill[black] (270:1) circle[radius=4pt];
  \end{tikzpicture}
  \caption{This graph shows the \emph{Fano plane}
    with seven points and seven lines such that every point is contained
    in exactly three lines and two lines intersect in exactly one point.
    Thus, interpreting points as items and lines as pools,
    the Fano plane describes a $(7,3,3)$-multipool
    which is not of the form provided by Theorem~\ref{thm:matrices}.}
  \label{fig:Fano}
\end{figure}%

In the subsequent results, we only rely on the multipool structure,
outlined in Definition~\ref{def:multipool}.
While Theorem~\ref{thm:matrices}
ensures that corresponding pooling matrices exist for particular~$n$, $q$, $m$,
there exist more, see Figure~\ref{fig:Fano}.
The following Theorems are valid beyond the restrictions imposed
by the matrices considered in Theorem~\ref{thm:matrices}.

\begin{theorem}\label{thm:sens_and_spec}
  Let $\rho$, $q$, $m$, $\nc$, $\pfp$ and $\pfn$ be given.
  Then, for any suitable~$n$, in any $(n,q,m)$-multipooling strategy
  with decoding by $\NCOMP(\nc)$, the sensitivity is
  \begin{equation}\label{eq:sens}
    \sens
    =\Prob(Z_j=1\given X_j=1)
    =\sum_{k=m-\nc}^{m}\binom mk
      \bigl(1-\pfn\gamma_1\bigr)^k
      \bigl(\pfn\gamma_1\bigr)^{m-k}
    \text,
   \end{equation}
  and the specificity is
  \begin{equation}\label{eq:spec}
    \spec
    =\Prob(Z_j=0\given X_j=0)
    =1-\sum_{k=m-\nc}^{m}\binom mk\bigl(1-\gamma_1\bigr)^k\gamma_1^{m-k}
    \text,
  \end{equation}
  where
  \begin{equation}\label{eq:gamma_1}
    \gamma_1
    =\Prob(Y_i=0\given X_j=0)
    =(1-\pfp)(1-(1-\pfn)\rho)^{q-1}\text.
  \end{equation}
\end{theorem}

\begin{remark}
  Note that when the pool size~$q$ and~$\nc$,
  that is the maximal number of negative pools an item can be in
  and still be flagged positive, are fixed,
  then the sensitivity is \emph{decreasing} in the multiplicity~$m$
  while the specificity is \emph{increasing} in~$m$.
  This follows from the inequality
  \begin{equation*}
    \sum_{k=m-\nc}^m\binom mk(1-x)^k x^{m-k}
    \le\sum_{k=m+1-\nc}^{m+1}\binom{m+1}k(1-x)^kx^{m+1-k}\text,
  \end{equation*}
  which is elementary upon noticing that the left and right hand side
  denote the probability of obtaining at least~$\nc$
  heads when flipping a $(1-x)$-biased coin~$m$ or~$m+1$ times, respectively.
  This can also be graphically observed in Figures~\ref{fig:sens}
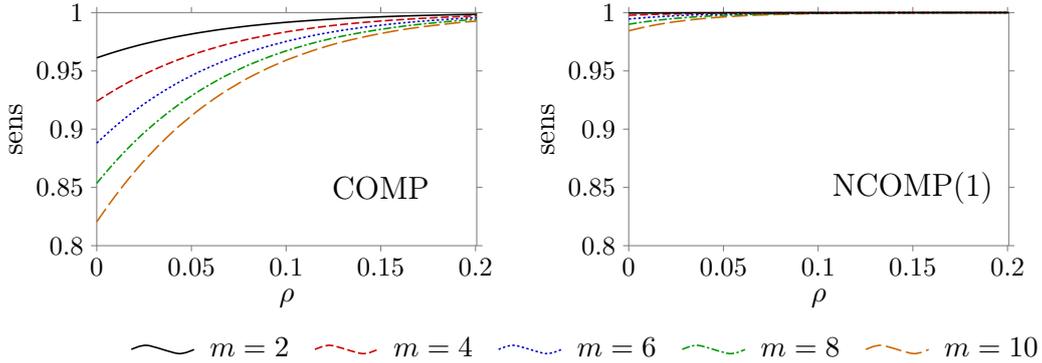
\begin{figure}[ht]
  \tikzsetnextfilename{sensitivity}
  \begin{tikzpicture}
    \datavisualization
    [ scientific axes
    , visualize as smooth line/.list={2,4,6,8,10}
    , x axis={ label={$\rho$}
             , ticks={ tick typesetter/.code=\twodecimals{##1} }
             }
    , y axis={ include value={.8, 1}
             , label={$\sens$}
             }
    , style sheet=strong colors
    , style sheet=vary dashing
    , legend={ anchor=north west
             , at={(.3,-1)}
             , max rows=1
             }
    , 2={label in legend={text={$m=2$}}}
    , 4={label in legend={text={$m=4$}}}
    , 6={label in legend={text={$m=6$}}}
    , 8={label in legend={text={$m=8$}}}
    , 10={label in legend={text={$m=10$}}}
    , data/format=function
    ]
    data {
      var set : {2,4,6,8,10};
      var x : interval [0:.201] samples 51;
      func y = sensofrho(\value x, \value{set}, 0, 16, .02, .02);
    }
    info {
      \node at (visualization cs: x=.15, y=.85) {COMP};
    }
    ;
    \begin{scope}[xshift=7cm]
      \datavisualization
      [ scientific axes
      , visualize as smooth line/.list={2,4,6,8,10}
      , x axis={ label={$\rho$}
               , ticks={ tick typesetter/.code=\twodecimals{##1} }
               }
      , y axis={ include value={.8, 1}
               , label={$\sens$}
               }
      , style sheet=strong colors
      , style sheet=vary dashing
      , data/format=function
      ]
      data {
        var set : {2,4,6,8,10};
        var x : interval [0:.201] samples 51;
        func y = sensofrho(\value x, \value{set}, 1, 16, .02, .02);
      }
      info {
        \node at (visualization cs: x=.15, y=.85) {NCOMP(1)};
      }
      ;
    \end{scope}
  \end{tikzpicture}
  \caption{The graphs show sensitivity of COMP and NCOMP(1) for pools of size
  $q=16$ and false positive and false negative probabilities $\pfp=\pfn=0.02$
  for multiplicities~$m\in\{2,4,6,8,10\}$
  against the prevalence~$\rho\in\Icc0{0.2}$.}
  \label{fig:sens}
\end{figure}%
  and~\ref{fig:spec}.
\begin{figure}[ht]
  \tikzsetnextfilename{specificity}
  \begin{tikzpicture}
    \datavisualization
    [ scientific axes
    , visualize as smooth line/.list={2,4,6,8,10}
    , x axis={ label={$\rho$}
             , ticks={ tick typesetter/.code=\twodecimals{##1} }
             }
    , y axis={ include value={0,1}
             , label={$\spec$}
             }
    , style sheet=strong colors
    , style sheet=vary dashing
    , legend={ anchor=north west
             , at={(.3,-1)}
             , max rows=1
             }
    , 2={label in legend={text={$m=2$}}}
    , 4={label in legend={text={$m=4$}}}
    , 6={label in legend={text={$m=6$}}}
    , 8={label in legend={text={$m=8$}}}
    , 10={label in legend={text={$m=10$}}}
    , data/format=function
    ]
    data {
      var set : {2,4,6,8,10};
      var x : interval [0:.3] samples 51;
      func y = specofrho(\value x, \value{set}, 0, 16, .02, .02);
    }
    info {
      \node at (visualization cs: x=.225, y=.8) {COMP};
    }
    ;
    \begin{scope}[xshift=7cm]
      \datavisualization
      [ scientific axes
      , visualize as smooth line/.list={2,4,6,8,10}
      , x axis={ label={$\rho$}
               , ticks={ tick typesetter/.code=\twodecimals{##1} }
               }
      , y axis={ include value={0,1}
               , label={$\spec$}
               }
      , style sheet=strong colors
      , style sheet=vary dashing
      , data/format=function
      ]
      data {
        var set : {2,4,6,8,10};
        var x : interval [0:.3] samples 51;
        func y = specofrho(\value x, \value{set}, 1, 16, .02, .02);
      }
      info {
        \node at (visualization cs: x=.225, y=.8) {NCOMP(1)};
      }
      ;
    \end{scope}
  \end{tikzpicture}
  \caption{The graphs show specificity of COMP and NCOMP(1)
    for pools of size $q=16$
    and false positive and false negative probabilities $\pfp=\pfn=0.02$.
    The specificity decreases when passing from COMP to NCOMP.
    Larger multiplicity~$m$ can mitigate this effect.}
  \label{fig:spec}
\end{figure}
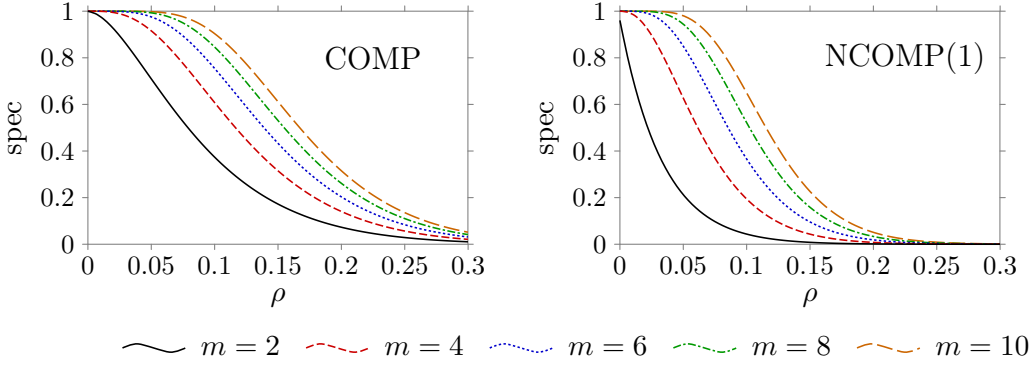%
  Figure~\ref{fig:sens}
  furthermore illustrates the error-correcting effect
  of the NCOMP algorithm in the presence of noise.
  We see that while the sensitivity always decreases with growing multiplicity~$m$,
  passing from COMP to NCOMP can mitigate this effect.
  In conclusion, a good strategy in the presence of a non-negligible
  false negative probability~$\pfn$ is to use NCOMP for high sensitivity
  and then boost the specificity by larger multiplicities.
\end{remark}

We can now also provide expressions for the probabilities
of Type~I and Type~II errors.

\begin{corollary}\label{cor:type_I_and_II}
  Let $\rho$, $q$, $m$, $\nc$, $\pfp$ and $\pfn$ be given.
  Then, for any suitable~$n$,
  in any $(n,q,m)$-multipooling strategy with decoding by $\NCOMP(\nc)$,
  the probability of Type I errors is
  \begin{equation}\label{eq:TypeI}
    \typeI
    =\left(1+\frac\rho{1-\rho}\cdot\frac{\sens}{1-\spec}\right)^{-1}
    \text,
  \end{equation}
  and the probability of Type II errors is
  \begin{equation}\label{eq:TypeII}
    \typeII
    =\left(1+\frac{1-\rho}\rho\cdot\frac{\spec}{1-\sens}\right)^{-1}
    \text.
  \end{equation}
\end{corollary}

\begin{remark}
  When decoding with the COMP decoder, \ie\ $\nc=0$,
  \eqref{eq:TypeII} simplifies to
  \begin{equation*}
    \Prob(X_j=0\given Z_j=1)
    =\left(1+
      \frac{\rho(1-\pfn\gamma_1)^m}
           {(1-\rho)(1-\gamma_1)^m}
     \right)^{-1}.
  \end{equation*}
  If we require the probability of Type I errors to be bounded by~$\epsilon>0$,
  then this condition can be solved for~$m$:
  \begin{equation}\label{eq:lower_bound_m}
    m\ge\log\left(\frac{1-\rho}\rho(\varepsilon^{-1}-1)\right)
      /\log\left(
      \frac{1-\pfn\gamma_1}
           {1-\gamma_1}
      \right)
  \end{equation}
  and provides a lower bound on the number of pools an item has to participate in.
  In the special case $\nc=0$ and $\pfp=\pfn=0$,
  we have $\gamma_1=(1-\rho)^{q-1}$ and recover the condition
  $\Prob(X_j=0\given Z_j=1)\le\varepsilon$ if and only if
  \begin{equation*}
    m\ge\log\left(\frac{1-\rho}\rho(\varepsilon^{-1}-1)\right)
      /\log\left(
      \frac1{1-(1-\rho)^{q-1}}
      \right)
  \end{equation*}
  from \cite{Taeufer-20}.
\end{remark}

\begin{remark}
  Let us study the Type I error probabilities in more detail.
  In Figure~\ref{fig:TypeI_noiseless_COMP},
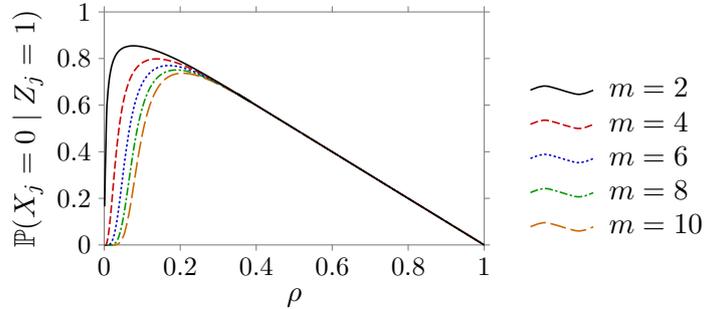
\begin{figure}[ht]
  \tikzsetnextfilename{typeICOMP00}
  \begin{tikzpicture}
    \datavisualization
    [ scientific axes
    , visualize as smooth line/.list={2,4,6,8,10}
    , x axis={ include value={0,1}
             , label={$\rho$}
             }
    , y axis={ include value={0,1}
             , label={$\typeI$}
             }
    , style sheet=strong colors
    , style sheet=vary dashing
    , legend=right
    , 2={label in legend={text={$m=2$}}}
    , 4={label in legend={text={$m=4$}}}
    , 6={label in legend={text={$m=6$}}}
    , 8={label in legend={text={$m=8$}}}
    , 10={label in legend={text={$m=10$}}}
    , data/format=function
    ]
    data {
      var set : {2,4,6,8,10};
      var x : {.001,.006,...,.051,.055,.065,...,0.4,1};
      func y = typeIofrho(\value x, \value{set}, 0, 16, 0, 0);
    }
    ;
  \end{tikzpicture}
  \label{subfig:typeICOMP00}
 \caption{Type I error rates for noiseless ($\pfp=\pfn=0$)
   testing with COMP for pools of size $q=16$
   and different multiplicities~$m$.}
 \label{fig:TypeI_noiseless_COMP}
\end{figure}%
  we see that in noiseless testing Type I errors
  emerge with growing prevalence rate and rapidly grow
  to approach the curve $f(\rho)=1-\rho$ for large~$\rho$.
  This is due to the fact that in this regime,
  a majority of pools will contain at least one positive item.
  Due to these \emph{combinatorial false positives},
  the whole test will become a useless oracle flagging every item positive.
  \par
  In Figure~\ref{fig:TypeI_noise}, we add noise.
\begin{figure}[ht]
  \tikzsetnextfilename{typeI202}
  \begin{tikzpicture}
    \datavisualization
    [ scientific axes
    , visualize as line/.list={2,4,6,8,10}
    , x axis={ include value={0,.1}
             , label={$\rho$}
             , ticks={ tick typesetter/.code=\twodecimals{##1} }
             }
    , y axis={ include value={0,1}
             , label={$\typeI$}
             }
    , style sheet=strong colors
    , style sheet=vary dashing
    , legend={ anchor=north west
             , at={(.3,-1)}
             , max rows=1
             }
    , 2={label in legend={text={$m=2$}}}
    , 4={label in legend={text={$m=4$}}}
    , 6={label in legend={text={$m=6$}}}
    , 8={label in legend={text={$m=8$}}}
    , 10={label in legend={text={$m=10$}}}
    , data/format=function
    ]
    data {
      var set : {2,4,6,8,10};
      var x : {0,.0005,...,.02995,.03,.032,...,.1};
      func y = typeIofrho(\value x, \value{set}, 0, 16, .2, .02);
    }
    info {
      \node at (visualization cs: x=.08, y=.2) {COMP};
    }
    ;
    \begin{scope}[xshift=7cm]
      \datavisualization
      [ scientific axes
      , visualize as line/.list={2,4,6,8,10}
      , x axis={ include value={0,.1}
               , label={$\rho$}
               , ticks={ tick typesetter/.code=\twodecimals{##1} }
               }
      , y axis={ include value={0,1}
               , label={$\typeI$}
               }
      , style sheet=strong colors
      , style sheet=vary dashing
      , data/format=function
      ]
      data {
        var set : {2,4,6,8,10};
        var x : {0,.0005,...,.02995,.03,.032,...,.1};
        func y = typeIofrho(\value x, \value{set}, 1, 16, .2, .02);
      }
      info {
        \node at (visualization cs: x=.075, y=.2) {NCOMP(1)};
      }
      ;
    \end{scope}
  \end{tikzpicture}
  \caption{Type I errors in the presence of noise ($\pfp=0.2$, $\pfn=0.02$)
    at small prevalence $\rho$ and for pools of size $q=16$.
    The false positive probability~$\pfp$
    has been chosen very high with $\pfp=0.2$
    in order to illustrate the screening paradox at small~$\rho$.
    For small prevalence~$\rho$,
    the graphs for NCOMP(1) with $m\in\{8,10\}$
    show numerical instablilty due to the smallness of $\Prob(Z_j=1)$.}
  \label{fig:TypeI_noise}
\end{figure}
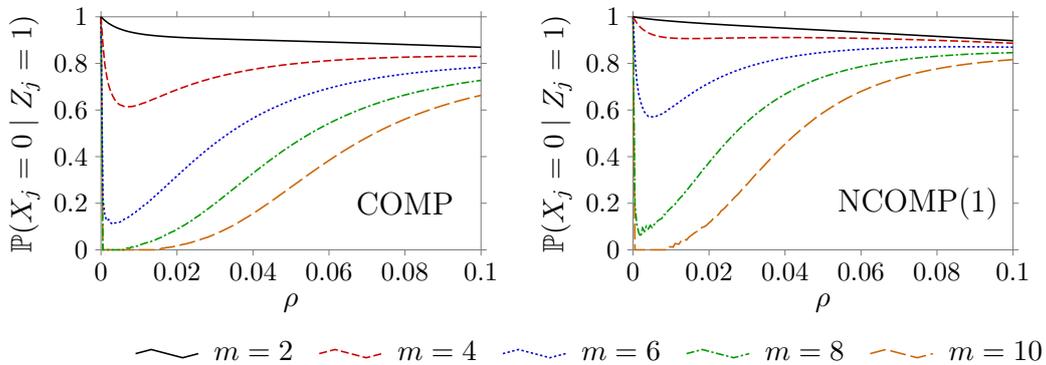%
  In the presence of a non-zero false positive probability~$\pfp$,
  we observe another phenomenon, namely the \emph{screening paradox}
  in which for small enough~$\rho$,
  true positives are so rare that they are dominated by false positives
  arising from noisy measurements.
  In any case, we observe that both types of false positives
  can be reduced by larger multiplicity~$m$.
  \par
  Finally, we also emphasize that larger pool sizes negatively impact
  Type I error probabilities, \cf\ Figure~\ref{fig:TypeI_varying_pool_size}.
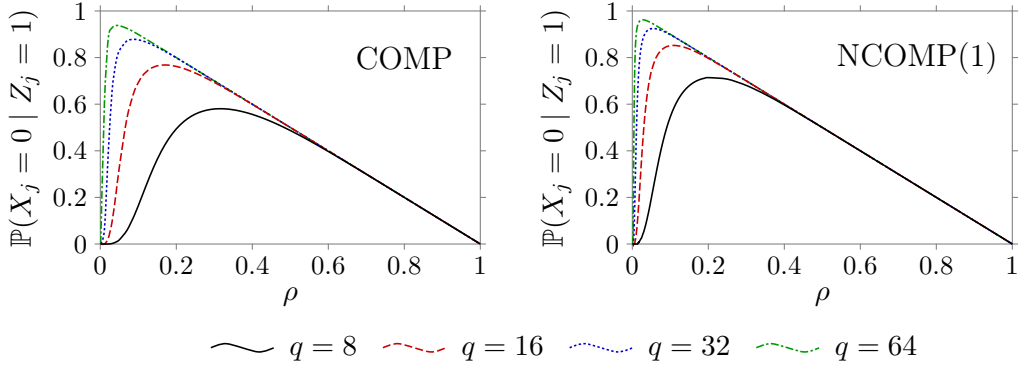
\begin{figure}[ht]
  \tikzsetnextfilename{typeI_q}
  \begin{tikzpicture}
    \datavisualization
    [ scientific axes
    , visualize as smooth line/.list={8,16,32,64}
    , x axis={ include value={0,1}
             , label={$\rho$}
             }
    , y axis={ include value={0,1}
             , label={$\typeI$}
             }
    , style sheet=strong colors
    , style sheet=vary dashing
    , legend={ anchor=north west
             , at={(1.3,-1)}
             , max rows=1
             }
    , 8={label in legend={text={$q=8$}}}
    , 16={label in legend={text={$q=16$}}}
    , 32={label in legend={text={$q=32$}}}
    , 64={label in legend={text={$q=64$}}}
    , data/format=function
    ]
    data {
      var set : {8,16,32,64};
      var x : {.001,.011,...,.041,.05,.07,...,.91,1};
      func y = typeIofrho(\value x, 6, 0, \value{set}, .02, .02);
    }
    info {
      \node at (visualization cs: x=.8, y=.8) {COMP};
    }
    ;
    \begin{scope}[xshift=7cm]
      \datavisualization
      [ scientific axes
      , visualize as smooth line/.list={8,16,32,64}
      , x axis={ include value={0,1}
               , label={$\rho$}
               }
      , y axis={ include value={0,1}
               , label={$\typeI$}
               }
      , style sheet=strong colors
      , style sheet=vary dashing
      , data/format=function
      ]
      data {
        var set : {4,8,16,32,64};
        var x : {.001,.002,...,.049,.05,.06,...,.19,.2,.25,...,.8,1};
        func y = typeIofrho(\value x, 6, 1, \value{set}, .02, .02);
      }
      info {
        \node at (visualization cs: x=.75, y=.8) {NCOMP(1)};
      }
      ;
    \end{scope}
  \end{tikzpicture}
  \caption{Type I errors for different pool sizes in the presence of noise
    ($\pfp=\pfn=0.02$) with multiplicity $m=6$
    and decoding with COMP and NCOMP(1).}
 \label{fig:TypeI_varying_pool_size}
\end{figure}%
  This again is remedied by larger multiplicity~$m$,
  where, in the light of~\eqref{eq:lower_bound_m},
  the necessary~$m$ grows only logarithmically with growing pool size
  such that after all the compression ratio
  rapidly improves with larger pool sizes.
\end{remark}

A consequence of Theorem~\ref{thm:sens_and_spec}
are expressions for the expected number of positive results~$T$,
of false positive, and false negative results $\Tfp$ and $\Tfn$
in screening strategies respectively.

\begin{corollary}\label{cor:expectation}
  Let $\rho$, $n$, $q$, $m$, $\nc$, $\pfp$ and $\pfn$ be given.
  Then, in any $(n,q,m)$-multipooling strategy
  with decoding by NCOMP with parameter~$\nc$,
  the expected number of all positive results~$T$ is
  \begin{equation}\label{eq:E[T]}
    \E[T]
    =n\bigl(\rho\cdot\sens+(1-\rho)(1-\spec)\bigr)\text,
  \end{equation}
  the expectation of the number~$\Tfp$ of all false positive results is
  \begin{equation}\label{eq:E[Tfp]}
    \E[\Tfp]
    =n(1-\rho)(1-\spec)\text,
  \end{equation}
  and the number~$\Tfn$ of all false negative results has expectation
  \begin{equation}\label{eq:E[Tfn]}
    \E[\Tfn]
    =n\rho(1-\sens)\text,
  \end{equation}
  where $\sens$ and $\spec$ are given in~\eqref{eq:sens}
  and~\eqref{eq:spec}.
\end{corollary}

\begin{remark}
  A noteworthy observation is that Corollary~\ref{cor:expectation}
  only relies on the three conditions in Definition~\ref{def:multipool},
  \ie\ constant row sum, constant column sum,
  scalar product between columns at most one.
  Thus, these conditions alone already determine
  the expected number of positives, false positives and false negatives.
  In particular, imposing further conditions on the pooling matrices
  will not reduce the expected number of false positives in COMP and NCOMP.
\end{remark}

\begin{remark}
  The expressions~\eqref{eq:E[T]}, \eqref{eq:E[Tfp]}, and~\eqref{eq:E[Tfn]}
  for $\E[T]$, $\E[\Tfp]$, and $\E[\Tfn]$
  are illustrated in Figure~\ref{fig:3ETs}
\begin{figure}[ht]
  \tikzsetnextfilename{3ETs}
  \begin{tikzpicture}
    \datavisualization
    [ scientific axes
    , visualize as smooth line/.list={ET,ETfp,ETfn}
    , x axis={ include value={0,1}
             , label={$\rho$}
             }
    , y axis={ include value={0,1}
             , label={$\E[T], \E[\Tfp], \E[\Tfn]$}
             }
    , style sheet=strong colors
    , style sheet=vary dashing
    , legend={ anchor=north west
             , at={(2.5,-1)}
             , max rows=1
             }
    , ET={label in legend={text={$\E[T]$}}}
    , ETfp={label in legend={text={$\E[\Tfp]$}}}
    , ETfn={label in legend={text={$\E[\Tfn]$}}}
    , data/format=function
    ]
    data [set = ET] {
      var x : {0,.01,...,.6,1};
      func y = ET(\value x, 8, 0, 16, .02, .02, 256);
    }
    data [set = ETfp] {
      var x : {0,.01,...,.6,1};
      func y = ETfp(\value x, 8, 0, 16, .02, .02, 256);
    }
    data [set = ETfn] {
      var x : {0,.01,...,.6,1};
      func y = ETfn(\value x, 8, 0, 16, .02, .02, 256);
    }
    info {
      \node at (visualization cs: x=.8, y=200) {COMP};
    }
    ;
    \begin{scope}[xshift=7cm]
      \datavisualization
      [ scientific axes
      , visualize as smooth line/.list={ET,ETfp,ETfn}
      , x axis={ include value={0,1}
               , label={$\rho$}
               }
      , y axis={ include value={0,1}
               , label={$\E[T], \E[\Tfp], \E[\Tfn]$}
               }
      , style sheet=strong colors
      , style sheet=vary dashing
      , data/format=function
      ]
      data [set = ET] {
        var x : {0,.01,...,.6,1};
        func y = ET(\value x, 8, 1, 16, .02, .02, 256);
      }
      data [set = ETfp] {
        var x : {0,.01,...,.6,1};
        func y = ETfp(\value x, 8, 0, 16, .02, .02, 256);
      }
      data [set = ETfn] {
        var x : {0,.01,...,.6,1};
        func y = ETfn(\value x, 8, 1, 16, .02, .02, 256);
      }
      info {
        \node at (visualization cs: x=.75, y=200) {NCOMP(1)};
      }
      ;
    \end{scope}
  \end{tikzpicture}
  \caption{Expected number of positives, false positives, and false negatives
    in a sample of~$256$ items for pools of size $q=16$ and multiplicity $m=8$
    in the presence of noise ($\pfp=\pfn=0.02$).}
  \label{fig:3ETs}
\end{figure}
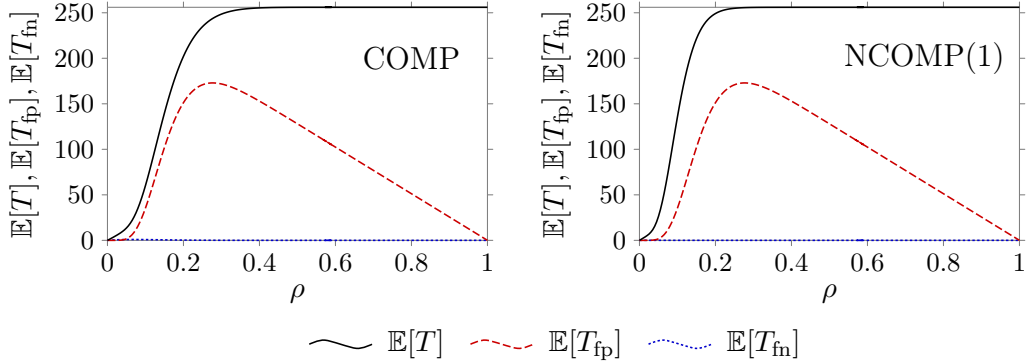%
  for~$256$ items in a multipool with pool size $q=16$,
  multiplicity $m=8$, and decoding with COMP and NCOMP(1), respectively.
  We observe again the phase transition from small~$\rho$,
  where non-adaptive testing works well, to moderate~$\rho$
  where essentially all~$256$ items are flagged positive.
  We also see that higher multiplicities help delaying this transition
  to higher~$\rho$, \cf\ Figure~\ref{fig:ETs}.
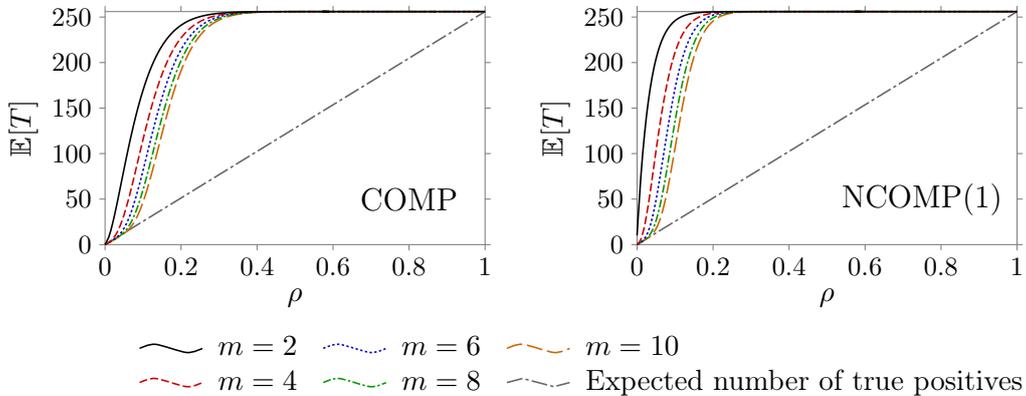
\begin{figure}[ht]
  \tikzsetnextfilename{ETs}
  \begin{tikzpicture}
    \datavisualization
    [ scientific axes
    , visualize as smooth line/.list={2,4,6,8,10}
    , visualize as line/.list={EBinnrho}
    , x axis={ include value={0,1}
             , label={$\rho$}
             , ticks={ tick typesetter/.code=\twodecimals{##1} }
             }
    , y axis={ include value={0,1}
             , label={$\E[T]$}
             }
    , style sheet=strong colors
    , style sheet=vary dashing
    , legend={ anchor=north west
             , at={(.3,-1)}
             , max rows=2
             }
    , 2={label in legend={text={$m=2$}}}
    , 4={label in legend={text={$m=4$}}}
    , 6={label in legend={text={$m=6$}}}
    , 8={label in legend={text={$m=8$}}}
    , 10={label in legend={text={$m=10$}}}
    , EBinnrho={label in legend={text={Expected number of true positives}}}
    , data/format=function
    ]
    data {
      var set : {2,4,6,8,10};
      var x : {0,.01,...,.6,1};
      func y = ET(\value x, \value{set}, 0, 16, .02, .02, 256);
    }
    data [set=EBinnrho] {
      var x : {0,1};
      func y = 256 * \value x;
    }
    info {
      \node at (visualization cs: x=.8, y=50) {COMP};
    }
    ;
    \begin{scope}[xshift=7cm]
      \datavisualization
      [ scientific axes
      , visualize as smooth line/.list={2,4,6,8,10,EBinnrho}
      , x axis={ include value={0,1}
               , label={$\rho$}
               , ticks={ tick typesetter/.code=\twodecimals{##1} }
               }
      , y axis={ include value={0,1}
               , label={$\E[T]$}
               }
      , style sheet=strong colors
      , style sheet=vary dashing
      , data/format=function
      ]
      data {
        var set : {2,4,6,8,10};
        var x : {0,.01,...,.6,1};
        func y = ET(\value x, \value{set}, 1, 16, .02, .02, 256);
      }
      data [set=EBinnrho] {
        var x : {0,1};
        func y = 256 * \value x;
      }
      info {
        \node at (visualization cs: x=.75, y=50) {NCOMP(1)};
      }
      ;
    \end{scope}
  \end{tikzpicture}
  \caption{Expected number of positive results in a sample of~$256$
    items for pools of size $q=16$
    and different multiplicities in the presence of noise ($\pfp=\pfn=0.02$).}
  \label{fig:ETs}
\end{figure}%
  If we focus on small~$\rho$,
  we see that the expected number of positives
  follows the expected number of true positives
  before it starts diverging.
  This happens later for larger multiplicities, \cf\ Figure~\ref{fig:ETszoom}.
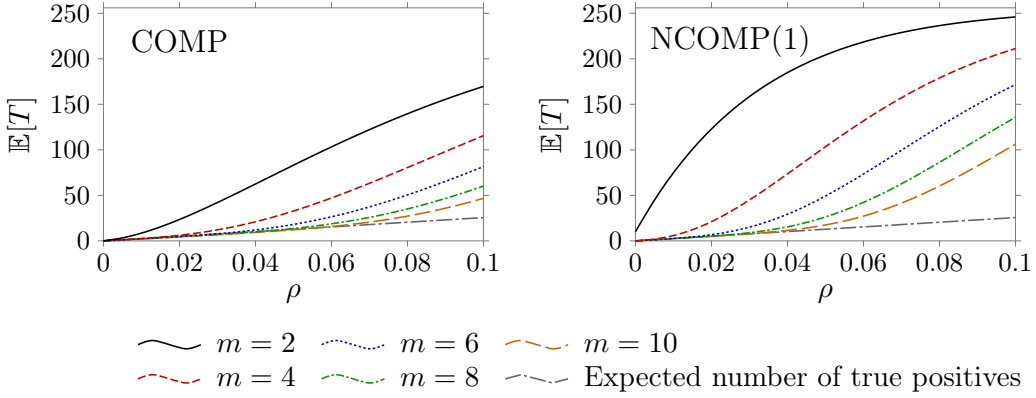
\begin{figure}
  \tikzsetnextfilename{ETszoom}
  \begin{tikzpicture}
    \datavisualization
    [ scientific axes
    , visualize as smooth line/.list={2,4,6,8,10}
    , visualize as line/.list={EBinnrho}
    , x axis={ include value={0,.1}
             , label={$\rho$}
             , ticks={ tick typesetter/.code=\twodecimals{##1} }
             }
    , y axis={ include value={0,256}
             , label={$\E[T]$}
             }
    , style sheet=strong colors
    , style sheet=vary dashing
    , legend={ anchor=north west
             , at={(.3,-1)}
             , max rows=2
             }
    , 2={label in legend={text={$m=2$}}}
    , 4={label in legend={text={$m=4$}}}
    , 6={label in legend={text={$m=6$}}}
    , 8={label in legend={text={$m=8$}}}
    , 10={label in legend={text={$m=10$}}}
    , EBinnrho={label in legend={text={Expected number of true positives}}}
    , data/format=function
    ]
    data {
      var set : {2,4,6,8,10};
      var x : interval [0:.1];
      func y = ET(\value x, \value{set}, 0, 16, .02, .02, 256);
    }
    data [set=EBinnrho] {
      var x : {0,.1};
      func y = 256 * \value x;
    }
    info {
      \node at (visualization cs: x=.02, y=220) {COMP};
    }
    ;
    \begin{scope}[xshift=7cm]
      \datavisualization
      [ scientific axes
      , visualize as smooth line/.list={2,4,6,8,10,EBinnrho}
      , x axis={ include value={0,.1}
               , label={$\rho$}
               , ticks={ tick typesetter/.code=\twodecimals{##1} }
               }
      , y axis={ include value={0,256}
               , label={$\E[T]$}
               }
      , style sheet=strong colors
      , style sheet=vary dashing
      , data/format=function
      ]
      data {
        var set : {2,4,6,8,10};
        var x : interval [0:.1];
        func y = ET(\value x, \value{set}, 1, 16, .02, .02, 256);
      }
      data [set=EBinnrho] {
        var x : {0,.1};
        func y = 256 * \value x;
      }
      info {
        \node at (visualization cs: x=.025, y=220) {NCOMP(1)};
      }
      ;
    \end{scope}
  \end{tikzpicture}
  \caption{Expected number of positive results at small prevalence $\rho$
    in a sample of~$256$ items for pools of size $q=16$ and different multiplicities in the presence of noise ($\pfp=\pfn=0.02$).
    With growing multiplicity $m$, $\E[T]$ approaches the expected number of true positives from above at small $\rho$ before diverging from it with increasing $\rho$.}
  \label{fig:ETszoom}
\end{figure}%
  Finally, Figure~\ref{fig:ETTfp_and_Tfn}
\begin{figure}[ht]
  \tikzsetnextfilename{ETfps}
  \begin{tikzpicture}
    \datavisualization
    [ scientific axes
    , visualize as smooth line/.list={2,4,6,8,10}
    , x axis={ include value={0,1}
             , label={$\rho$}
             , ticks={ tick typesetter/.code=\twodecimals{##1} }
             }
    , y axis={ include value={0,230}
             , label={$\E[\Tfp]$}
             }
    , style sheet=strong colors
    , style sheet=vary dashing
    , data/format=function
    ]
    data {
      var set : {2,4,6,8,10};
      var x : {0,.01,...,.6,1};
      func y = ETfp(\value x, \value{set}, 0, 16, .02, .02, 256);
    }
    info {
      \node at (visualization cs: x=.8, y=180) {COMP};
    }
    ;
    \begin{scope}[xshift=7cm]
      \datavisualization
      [ scientific axes
      , visualize as smooth line/.list={2,4,6,8,10}
      , x axis={ include value={0,1}
               , label={$\rho$}
               , ticks={ tick typesetter/.code=\twodecimals{##1} }
               }
      , y axis={ include value={0,230}
               , label={$\E[\Tfp]$}
               }
      , style sheet=strong colors
      , style sheet=vary dashing
      , data/format=function
      ]
      data {
        var set : {2,4,6,8,10};
        var x : {0,.01,...,.6,1};
        func y = ETfp(\value x, \value{set}, 1, 16, .02, .02, 256);
      }
      info {
        \node at (visualization cs: x=.75, y=180) {NCOMP(1)};
      }
      ;
    \end{scope}
  \end{tikzpicture}

  \tikzsetnextfilename{ETfns}
  \begin{tikzpicture}
    \datavisualization
    [ scientific axes
    , visualize as smooth line/.list={2,4,6,8,10}
    , x axis={ include value={0,1}
             , label={$\rho$}
             , ticks={ tick typesetter/.code=\twodecimals{##1} }
             }
    , y axis={ include value={0,1.2}
             , label={$\E[\Tfn]$}
             }
    , style sheet=strong colors
    , style sheet=vary dashing
    , legend={ anchor=north west
             , at={(.3,-1)}
             , max rows=1
             }
    , 2={label in legend={text={$m=2$}}}
    , 4={label in legend={text={$m=4$}}}
    , 6={label in legend={text={$m=6$}}}
    , 8={label in legend={text={$m=8$}}}
    , 10={label in legend={text={$m=10$}}}
    , data/format=function
    ]
    data {
      var set : {2,4,6,8,10};
      var x : {0,.01,...,.6,1};
      func y = ETfn(\value x, \value{set}, 0, 16, .02, .02, 256);
    }
    info {
      \node at (visualization cs: x=.8, y=.25) {COMP};
    }
    ;
    \begin{scope}[xshift=7cm]
      \datavisualization
      [ scientific axes
      , visualize as smooth line/.list={2,4,6,8,10}
      , x axis={ include value={0,1}
               , label={$\rho$}
               , ticks={ tick typesetter/.code=\twodecimals{##1} }
               }
      , y axis={ include value={0,1.2}
               , label={$\E[\Tfn]$}
               }
      , style sheet=strong colors
      , style sheet=vary dashing
      , data/format=function
      ]
      data {
        var set : {2,4,6,8,10};
        var x : {0,.01,...,.6,1};
        func y = ETfn(\value x, \value{set}, 1, 16, .02, .02, 256);
      }
      info {
        \node at (visualization cs: x=.75, y=.25) {NCOMP(1)};
      }
      ;
    \end{scope}
  \end{tikzpicture}
  \caption{Expected number of false positives and false negatives
    in a sample of~$256$ items for pools of size $q=16$
    in the presence of noise ($\pfp=\pfn=0.02$) for COMP and NCOMP(1).
    Note that passing from COMP to NCOMP(1)
    slashes the false negative probability.}
  \label{fig:ETTfp_and_Tfn}
\end{figure}
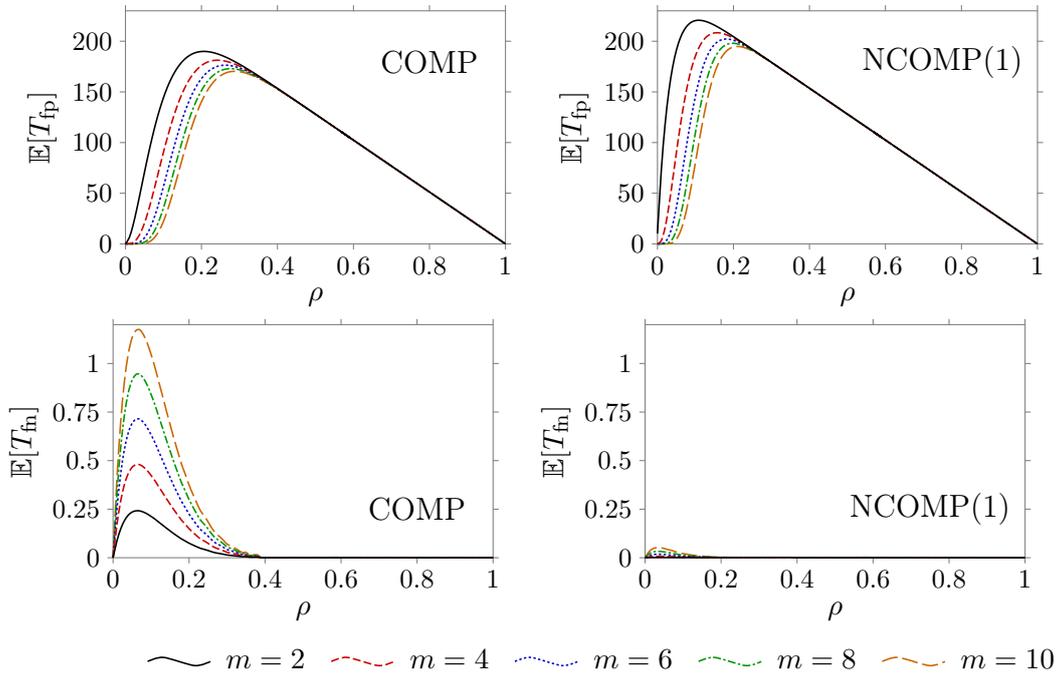%
  illustrates the interplay
  between the expected number of false positives and COMP and NCOMP(1):
  Passing from COMP to NCOMP(1) will increase
  the expected number of false positives,
  but in the presence of a non-negligible false negative probability,
  NCOMP(1) will also slash the expected number of false negatives close to~$0$.
\end{remark}

In order to better understand the random variables~$T$ and~$\Tfp$,
we provide bounds on their variance.
We restrict ourselves to the case $\nc=0$, \ie\ decoding by COMP,
and $\pfn=\pfp=0$, that is noiseless testing.
While it is possible to fix a particular matrix,
and write down analytic expressions for this variance,
we provide here a universal estimate on the variances which only relies
on the multipool structure of Definition~\ref{def:multipool}.
Its proof uses the Efron--Stein estimate and
is given in Section~\ref{sec:Efron-Stein}.

\begin{theorem}\label{thm:variance}
  If $\pfp=\pfn=0$ (noiseless testing) and $\nc=0$ (decoding by COMP),
  in any $(n,q,m)$-multipool strategy, we have for
  the expectations of the number~$T$ of positive results
  and of the number~$\Tfp$ of false positive results:
  \marginpar{\red{$\helper$?}}%
  \begin{align}
    \Var[T]&
    \le nmq\rho(1-\rho)\bigl(1-\helper^m+m(q-1)(1-\rho)^{q-1}\helper^{m-1}\bigr)
    \text,\label{eq:Variance_T}\\
    \Var[\Tfp]&
    \le nmq\rho(1-\rho)\bigl(\helper^m+m(q-1)(1-\rho)^{q-1}\helper^{m-1}\bigr)
    \text,\label{eq:Variance_Tfp}
  \end{align}
  where $\helper=1-(1-\rho)^{q-1}$.
\end{theorem}

\begin{remark}
  Figure~\ref{fig:variance} illustrates the bounds~\eqref{eq:Variance_T}
  and~\eqref{eq:Variance_Tfp} on the variances.
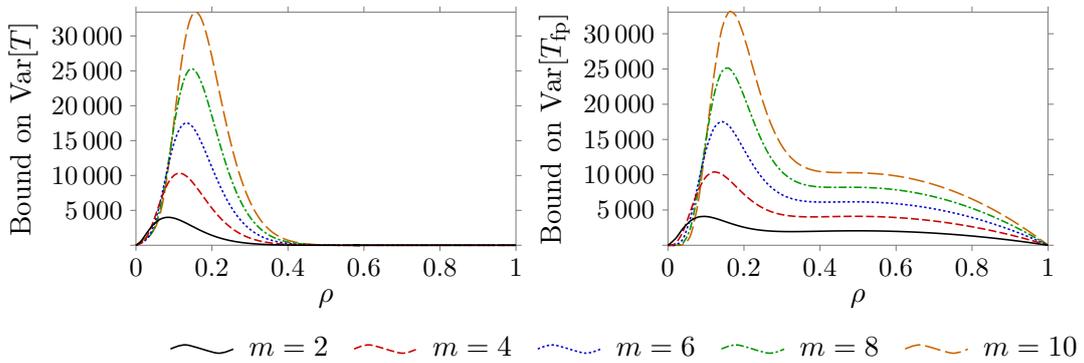
\begin{figure}[ht]
  \tikzsetnextfilename{Varbound}
  \begin{tikzpicture}
    \datavisualization
    [ scientific axes
    , visualize as smooth line/.list={2,4,6,8,10}
    , x axis={ label={$\rho$} }
    , y axis={ include value={0,1}
             , label={Bound on $\Var[T]$}
             , ticks={ tick suffix=$\,000$
                     , major={ no tick text at=0 }
                     }
             }
    , style sheet=strong colors
    , style sheet=vary dashing
    , legend={ anchor=north west
             , at={(.3,-1)}
             , max rows=1
             }
    , 2={label in legend={text={$m=2$}}}
    , 4={label in legend={text={$m=4$}}}
    , 6={label in legend={text={$m=6$}}}
    , 8={label in legend={text={$m=8$}}}
    , 10={label in legend={text={$m=10$}}}
    , data/format=function
    ]
    data {
      var set : {2,4,6,8,10};
      var x : {0,.01,...,.6,1};
      func y = VarTupper(\value x, \value{set}, 16, .256);
    }
    ;
    \begin{scope}[xshift=7cm]
      \datavisualization
      [ scientific axes
      , visualize as smooth line/.list={2,4,6,8,10}
      , x axis={label={$\rho$}}
      , y axis={ include value={0,1}
               , label={Bound on $\Var[\Tfp]$}
               , ticks={ tick suffix=$\,000$
                       , major={ no tick text at=0 }
                       }
               }
      , style sheet=strong colors
      , style sheet=vary dashing
      , data/format=function
      ]
      data {
        var set : {2,4,6,8,10};
        var x : interval [0:1] samples 51;
        func y = VarTfpupper(\value x, \value{set}, 16, .256);
      }
      ;
    \end{scope}
  \end{tikzpicture}
  \caption{The graphs illustrate the bounds~\eqref{eq:Variance_T}
    and~\eqref{eq:Variance_Tfp} on the variance of~$T$ and~$\Tfp$
    for noiseless testing in pools of size~$16$
    taken from a sample with~$256$ items
    for different multiplicities~$m$.}
\label{fig:variance}
\end{figure}%
  Again, we see the transition from low variance at small prevalence~$\rho$
  to huge variance at medium~$\rho$
  until the decoding will flag essentially everyone positive,
  thus for large~$\rho$, it will become a useless prediction.
  In this regime,~$T$ will have zero variance, and the variance~$\Tfp$
  is essentially the variance of the binomial distribution of uninfected items.

  More interesting observations can be made by zooming in on small~$\rho$,
  see Figure~\ref{fig:Varzoom}.
\begin{figure}[ht]
  \tikzsetnextfilename{Varboundzoom}
  \begin{tikzpicture}
    \datavisualization
    [ scientific axes
    , visualize as smooth line/.list={2,4,6,8,10,bin}
    , x axis={ label={$\rho$}
             , ticks={ step=.02
                     , minor steps between steps=1
                     , tick typesetter/.code=\twodecimals{##1}
                     }
             }
    , y axis={ include value={0,9}
             , label={Bound on $\Var[T]$}
             , ticks={ tick suffix=$\,000$
                     , major={ no tick text at=0 }
                     }
             }
    , style sheet=strong colors
    , style sheet=vary dashing
    , legend={ anchor=north west
             , at={(.3,-1)}
             , max rows=2
             }
    , 2={label in legend={text={$m=2$}}}
    , 4={label in legend={text={$m=4$}}}
    , 6={label in legend={text={$m=6$}}}
    , 8={label in legend={text={$m=8$}}}
    , 10={label in legend={text={$m=10$}}}
    , bin={label in legend={text={$256\rho(1-\rho)$}}}
    , data/format=function
    ]
    data {
      var set : {2,4,6,8,10};
      var x : interval [0:.075] samples 51;
      func y = VarTupper(\value x, \value{set}, 16, .256);
    }
    data [set=bin] {
      var x : interval [0:.075];
      func y = 256*\value x*(1-\value x)/1000;
    }
    ;
    \begin{scope}[xshift=7cm]
      \datavisualization
      [ scientific axes
      , visualize as smooth line/.list={2,4,6,8,10}
      , x axis={ label={$\rho$}
               , ticks={ step=.02
                       , minor steps between steps=1
                       , tick typesetter/.code=\twodecimals{##1}
                       }
               }
      , y axis={ include value={0,9}
               , label={Bound on $\Var[\Tfp]$}
               , ticks={ tick suffix=$\,000$
                       , major={ no tick text at=0 }
                       }
               }
      , style sheet=strong colors
      , style sheet=vary dashing
      , data/format=function
      ]
      data {
        var set : {2,4,6,8,10};
        var x : interval [0:0.075] samples 51;
        func y = VarTfpupper(\value x, \value{set}, 16, .256);
      }
      ;
    \end{scope}
  \end{tikzpicture}
  \caption{Bounds on~$\Var[T]$ and~$\Var[\Tfp]$ for small prevalence~$\rho$,
    noiseless testing, and pools of size $q=16$
    taken from a sample with~$256$ items.
  }
  \label{fig:Varzoom}
\end{figure}
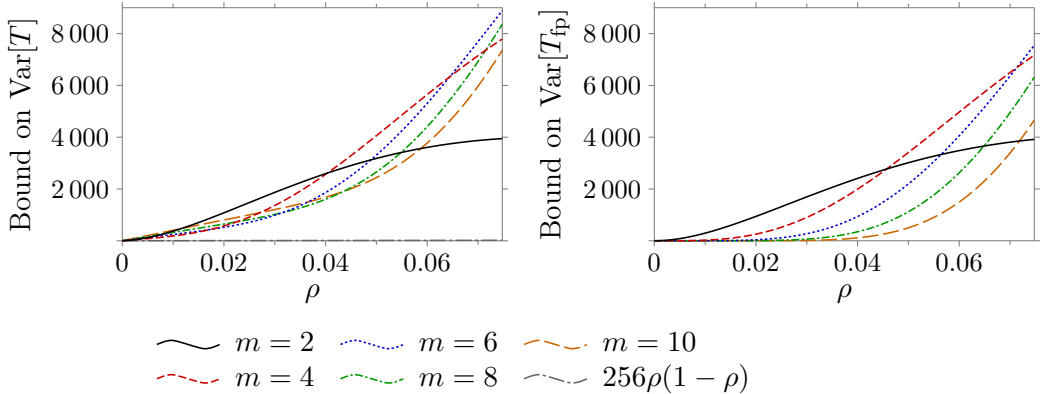%
  Firstly, we see that our bound on~$\Var[T]$ exceeds $256\rho(1-\rho)$,
  the variance of the true number of infected items,
  even in a regime where the expected numbers are rather close.
  Since this difference exceeds our upper bound on~$\Var[\Tfp]$,
  we conclude that our bound on~$\Var[T]$ is far from sharp.
  However, the bound on~$\Var[\Tfp]$
  seems to be closer to optimal since it remains near zero for small~$\rho$
  and then grows with a steeper incline,
  carrying features of a phase transition.
  This effect becomes more prominent with increasing multiplicity~$m$,
  \cf\ Figure~\ref{fig:Varzoom}.

  One could argue that a very large variance is preferable
  to a moderately large one.
  Indeed, in the latter case,
  one will have the odd run with a false positive result
  which is hard to identify,
  whereas in the first case, false positive results will come in rare batches
  with unusually large portions of positives.
  Such a pattern could be flagged as a ``graceful failure''
  of the testing strategy which points to an outbreak
  without identifying the infected individuals.
  For this purpose, one would like to have a clear phase transition
  between perfect recovery and graceful failure which,
  considering Figure~\ref{fig:Varzoom}, requires large multiplicity~$m$.
\end{remark}

\begin{remark}
  Let us conclude our remarks with some information theoretic considerations.
  For the sake of the argument, we consider the case $n=q^2$.
  One can think of two thresholds for~$\rho$
  above which non-adaptive group testing might break down.
  The first one is the \emph{combinatorial} or \emph{disjuncness threshold}
  \begin{equation*}
    \rhodisj
    =\frac{m-1}{n}
    =\frac{m-1}{q^2}\text.
  \end{equation*}
  If the prevalence~$\rho$ exceeds~$\rhodisj$,
  then in average more than half of all test runs
  will be confronted with a number of infected items for which
  the pooling matrix cannot guarantee perfect identification any more.
  However, we see in Figure~\ref{fig:overclock}
  \begin{figure}[ht]
  \tikzsetnextfilename{overclock}
  \begin{tikzpicture}[scale=2]
    \datavisualization
    [ scientific axes
    , visualize as line=EBinnrho
    , visualize as smooth line/.list={3,4,6}
    , visualize as line=5percent
    , x axis=
      { include value={0}
      , label={$\rho$}
      , ticks and grid=
        { major=
          { also at=
            { (2/256) as [{style={red!80!black,thin,densely dashed}, tick text padding=1ex}]
                $\rhodisj$
            , (3/256) as [{style={blue!80!black,thin,densely dotted}}]
                $\rhodisj$
            , (5/256) as [{style={green!60!black,thin,densely dash dot}, tick text padding=1ex}]
                $\rhodisj$
            , (.027171) as [{style={red!80!black,thin,densely dashed}, tick text padding=1ex}]
                $\rhoinfo$
            , (.041692) as [{style={blue!80!black,thin,densely dotted}}]
                $\rhoinfo$
            , (.07245) as [{style={green!60!black,thin,densely dash dot}, tick text padding=1ex}]
                $\rhoinfo$
            }
          }
        }
      }
    ,y axis={ include value={0}
             , label={$\E[\Tfp]$}
             }
    , style sheet=strong colors
    , style sheet=vary dashing
    , data/format=function
    ]
    data [set=EBinnrho] {
      var x : {0,.12};
      func y = 256 * \value x;
    }
    data [set=5percent] {
      var x : {0,.12};
      func y = 25.6 * \value x;
    }
    data [set=3] {
      var x : interval [0:.03];
      func y = ETfp(\value x, \value{set}, 0, 16, .02, .02, 256);
    }
    data [set=4] {
      var x : interval [0:.05];
      func y = ETfp(\value x, \value{set}, 0, 16, .02, .02, 256);
    }
    data [set=6] {
      var x : interval [0:.075];
      func y = ETfp(\value x, \value{set}, 0, 16, .02, .02, 256);
    }
    ;
  \end{tikzpicture}
  \caption{The expected number of false positives and the thresholds
    $\rho_{\mathrm{disj}} = \frac{m-1}{q^2}$
    and $\rho_{\mathrm{ent}}=h^{-1}(\frac{m}{q})$
    for $n=256$ items in pools of size $q=16$
    and noisy testing ($\pfp=\pfn=0.02$).}
  \label{fig:overclock}
\end{figure}
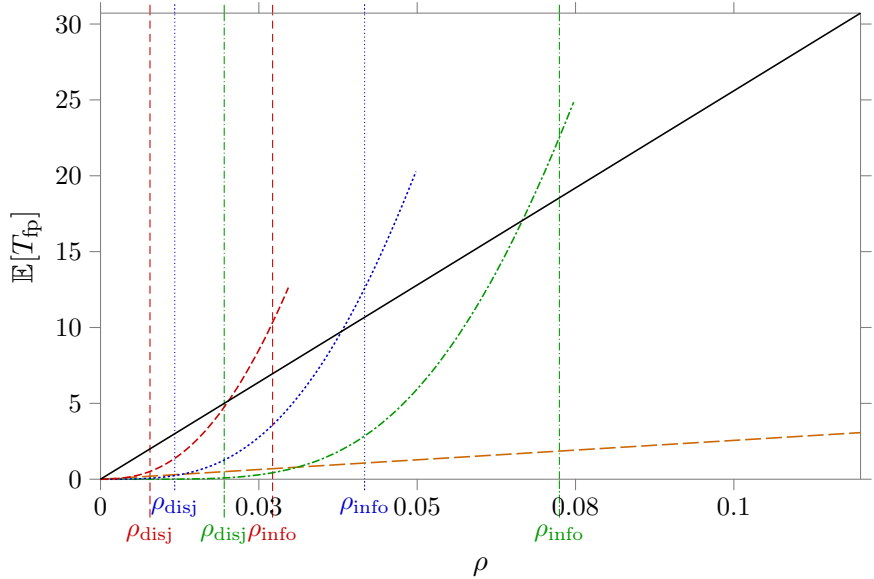%
  that identification still works rather well for moderately higher prevalences,
  \ie\ the expected number of false positives remains small
  and the good news is that it is perfectly possible
  to overclock COMP beyond~$\rhodisj$.
  The second threshold is the \emph{information theoretic threshold}
  \begin{equation*}
    \rhoinfo
    =h^{-1}\left(\frac mq\right)\text,
  \end{equation*}
  where $h(x)=x\log_2(x)+(1-x)\log_2(1-x)$
  denotes the entropy function,
  and we take its inverse on the interval $\Ioo0{1/2}$ where it is monotone.
  For prevalences~$\rho$ above~$\rhoinfo$,
  the average information contained in a binary string of length~$n$
  with a portion~$\rho$ of ones will exceed the possible information
  that can be encoded in a binary string of length $m\cdot q$,
  \ie\ in the number of pools,
  thus imposing a hard limit on the maximal prevalence~$\rho$
  for all possible compression ratios.
  We see that the number of false positives
  starts to blow up near~$\rhoinfo$.
  However, we can observe in the logarithmic plot
  in Figure~\ref{fig:overclocklog}
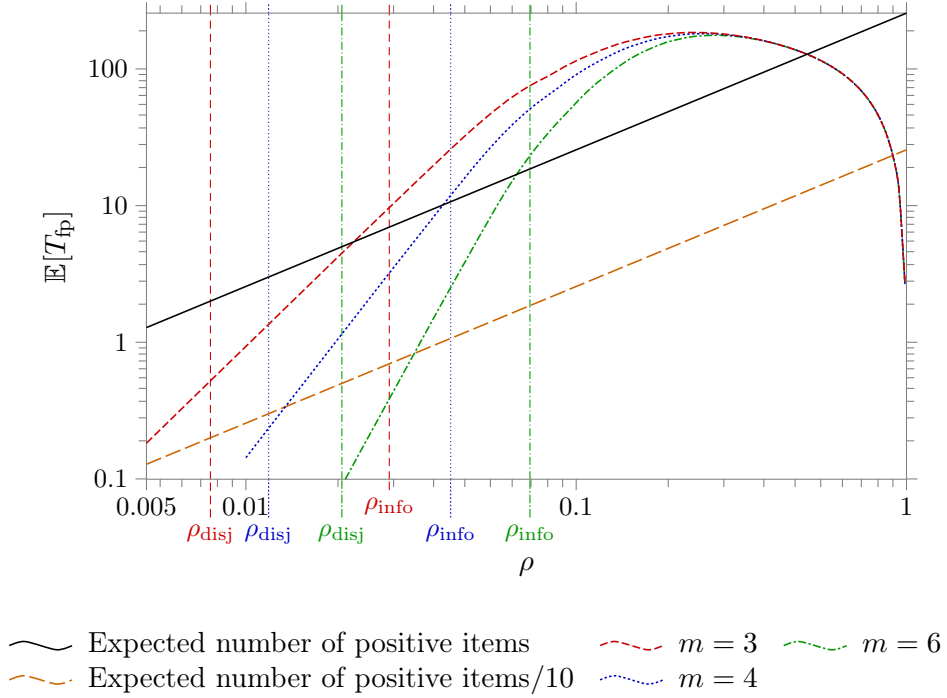
\begin{figure}[ht]
  \tikzsetnextfilename{overclocklog}
  \begin{tikzpicture}[scale=2]
    \datavisualization
    [ scientific axes
    , visualize as smooth line/.list={EBinnrho,3,4,6,5percent}
    , x axis=
      { label={$\rho$}
      , logarithmic
      , ticks={ major={ also at={ .005 as {$0.005$}
                                }
                      }
              , minor steps between steps
              }
      , ticks and grid=
        { major=
          { also at=
            { (2/256) as [{style={red!80!black,thin,densely dashed}, tick text padding=1ex}]
                $\rhodisj$
            , (3/256) as [{style={blue!80!black,thin,densely dotted}, tick text padding=1ex}]
                $\rhodisj$
            , (5/256) as [{style={green!60!black,thin,densely dash dot}, tick text padding=1ex}]
                $\rhodisj$
            , (.027171) as [{style={red!80!black,thin,densely dashed}}]
                $\rhoinfo$
            , (.041692) as [{style={blue!80!black,thin,densely dotted}, tick text padding=1ex}]
                $\rhoinfo$
            , (.07245) as [{style={green!60!black,thin,densely dash dot}, tick text padding=1ex}]
                $\rhoinfo$
            }
          }
        }
      }
    , y axis={ label={$\E[\Tfp]$}
             , logarithmic
             , ticks={ minor steps between steps }
             }
    , data/format=function
    , style sheet=strong colors
    , style sheet=vary dashing
    , legend={ below 
             , max rows=2
             }
    , EBinnrho={label in legend={text={Expected number of positive items}}}
    , 5percent={label in legend={text={$\text{Expected number of positive items}/10$}}}
    , 3={label in legend={text={$m=3$}}}
    , 4={label in legend={text={$m=4$}}}
    , 6={label in legend={text={$m=6$}}}
    ]
    data [set=EBinnrho] {
      var x : interval [.005:1];
      func y = 256 * \value x;
    }
    data [set=5percent] {
      var x : interval [.005:1];
      func y = 25.6 * \value x;
    }
    data [set=3] {
      var x : interval [.005:.99];
      func y = ETfp(\value x, \value{set}, 0, 16, .02, .02, 256);
    }
    data [set=4] {
      var x : interval [.01:.99];
      func y = ETfp(\value x, \value{set}, 0, 16, .02, .02, 256);
    }
    data [set=6] {
      var x : interval [.02:.99];
      func y = ETfp(\value x, \value{set}, 0, 16, .02, .02, 256);
    }
    ;
  \end{tikzpicture}
  \caption{Doubly logarithmic plot of the expected number of false positives
    and the thresholds $\rho_{\mathrm{disj}} = \frac{m-1}{q^2}$
    and $\rho_{\mathrm{ent}} = h^{-1}(\frac{m}{q})$ for $n=256$
    items in pools of size $q=16$ and noisy testing ($\pfp=\pfn=0.02$).
    At the information theoretic maximal prevalence $\rho_{\mathrm{ent}}$,
    the expected number of false negatives is well less than
    one order of magnitude below the expected true number of positives.}
  \label{fig:overclocklog}
\end{figure}%
  that even at the information theoretic maximal prevalence~$\rhoinfo$,
  the false positives are far less than one magnitude above the true positives.
  This means that while we would not recommend running
  non-adaptive multipool testing at prevalences
  near the information theoretic maximal prevalence~$\rhoinfo$,
  \emph{even in this extreme regime,
    the algorithm produces results which are useful
    as a first stage in a screening strategy}.
\end{remark}


\section{Sensitivity, specificity, and error probabilities for COMP and NCOMP}
  \label{sec:calculations}

In this section we prove Theorems~\ref{thm:sens_and_spec}
and~\ref{cor:type_I_and_II}.
We start with the following lemma.

\begin{lemma}\label{lemma:P(Y=0|X)}
  Fix $k\in\{0,\dotsc,q\}$, a pool~$\Pi_i$, and $j_1,\dotsc,j_k\in\Pi_i$.
  Then
  \begin{align*}
    \Prob(Y_i=0\given X_{j_1},\dotsc,X_{j_k})
    =\gamma_k\pfn^{X_{j_1}+\dotsb+X_{j_k}}
  \end{align*}
  and
  \begin{align*}
    \Prob(Y_i=1\given X_{j_1},\dotsc,X_{j_k})
    =
    1 - \gamma_k\pfn^{X_{j_1}+\dotsb+X_{j_k}}
  \end{align*}
  with $\gamma_k:=\Prob(Y_i=0\given X_{j_1}=\dotsb=X_{j_k}=0)
                 =(1-\pfp)(1-(1-\pfn)\rho)^{q-k}$.
\end{lemma}
\begin{proof}
  We condition on the number~$s$ of positive items in~$\Pi_i$
  besides the one we condition on, apply~\eqref{eq:error_model},
  and then simplify using Newton's theorem:
  \begin{align*}&
    \Prob(Y_i=0\given X_{j_1},\dotsc,X_{j_k})\\&
    =\E\Bigl[\Prob\Bigl(Y_i=0\Bigm|(AX)_i-\sum_{\ell=1}^kX_{j_\ell}\Bigr)
      \Bigm|X_{j_1},\dotsc,X_{j_k}\Bigr]\\&
    =\sum_{s=0}^{q-k}(1-\pfp)\pfn^{s+\sum_{\ell=1}^kX_{j_\ell}}
      \binom{q-k}s\rho^s(1-\rho)^{q-k-s}\\&
    =(1-\pfp)(\pfn\rho+1-\rho)^{q-k}\pfn^{\sum_{\ell=1}^kX_{j_\ell}}
    =\gamma_k\pfn^{\sum_{\ell=1}^kX_{j_\ell}}
    \text.
  \end{align*}
  Choosing $X_{j_1}=\dotsb=X_{j_k}=0$,
  the statement on~$\gamma_k$ follows.
\end{proof}

Lemma~\ref{lemma:P(Y=0|X)} allows us to calculate
the sensitivity and specificity for Theorem~\ref{thm:sens_and_spec}.

\begin{proof}[Proof of Theorem~\ref{thm:sens_and_spec}]
  For item~$i$ to test positive, at most~$\nc$ pools containing item~$i$
  can test positive.
  The number of pools with positive result for item~$j$ is $(A^\top Y)_j$.
  Therefore, we have by conditional independence and Lemma~\ref{lemma:P(Y=0|X)}
  \begin{align}
    \sens&
    =\Prob(Z_j=1\given X_j=1)
    =\Prob\bigl((A^\top Y)_j\ge m-\nc\given X_j=1\bigr)\notag\\&
    =\sum_{k=m-\nc}^m\Prob((A^\top Y)_j=k\given X_j=1)\notag\\&
    =\sum_{k=m-\nc}^m\binom mk\Prob(Y_i=1\given X_j=1)^k
      \Prob(Y_i=0\given X_j=1)^{m-k}\notag\\&
    =\sum_{k=m-\nc}^m\binom mk\bigl(1-\pfn \gamma_1\bigr)^k
      \bigl(\pfn \gamma_1\bigr)^{m-k}
    \label{eq:P(Z=1|X=1)}
  \end{align}
  which shows~\eqref{eq:sens}.
  Identity~\eqref{eq:spec} follows analogously.
\end{proof}

We apply Theorem~\ref{thm:sens_and_spec}
to calculate the probability of Type~I and Type~II errors
in Corollary~\ref{cor:type_I_and_II}.

\begin{proof}[Proof of Corollary~\ref{cor:type_I_and_II}]
  Using Bayes formula, we have
  \begin{multline}\label{eq:bayes}
    \Prob(X_j=0\given Z_j=1)
    =\frac{\Prob(X_j=0)\Prob(Z_j=1\given X_j=0)}{\Prob(Z_j=1)}\\
    =\frac{\Prob(X_j=0)\Prob(Z_j=1\given X_j=0)}
          {\Prob(X_j=0)\Prob(Z_j=1\given X_j=0)
          +\Prob(X_j=1)\Prob(Z_j=1\given X_j=1)}\\
    =\left(1+\frac{\Prob(X_j=1)\Prob(Z_j=1\given X_j=1)}
                  {\Prob(X_j=0)\Prob(Z_j=1\given X_j=0)}
     \right)^{-1}
    =\left(1+\frac{\rho}{1-\rho}\cdot\frac{\sens}{(1-\spec)}\right)^{-1}
    \text.
  \end{multline}
  This shows \eqref{eq:TypeI}.
  Identity~\eqref{eq:TypeII} is established analogously.
\end{proof}

Finally, we derive the expectations given in Corollary~\ref{cor:expectation}.
\begin{proof}[Proof of Corollary~\ref{cor:expectation}]
  We calculate
  \begin{align*}
    \E[T]&
    =\E\Bigl[\sum\nolimits_{j=1}^n\ifu{\{Z_j=1\}}\Bigr]
    =n\Prob(Z_1=1)\\&
    =n\sum\nolimits_{k=0}^1\Prob(X_1=k)\Prob(Z_1=1\given X_1=k)\\&
    =n\bigl(\rho\cdot\sens+(1-\rho)(1-\spec)\bigr)\text,
  \end{align*}
  \begin{align*}
    \E[\Tfp]&
    =\E\Bigl[\sum\nolimits_{j=1}^n\ifu{\{X_j=0,Z_j=1\}}\Bigr]
    =n\Prob(X_1=0,Z_1=1)\\&
    =n\Prob(X_1=0)\Prob(Z_1=1\given X_1=0)
    =n(1-\rho)(1-\spec)\text,
  \end{align*}
  and analogously
  \begin{equation*}
    \E[\Tfn]
    =n\Prob(X_1=1)\Prob(Z_1=0\given X_1=1)
    =n\rho(1-\sens)\text.\qedhere
  \end{equation*}
\end{proof}

\section{Bounding the variance}
  \label{sec:Efron-Stein}

In this section, we prove Theorem~\ref{thm:variance}.
We are in the situation $\nc=0$, \ie\ decoding by COMP,
and $\pfp=\pfn=0$, which means noiseless testing.
For an $(n,q,m)$-multipool matrix we want to estimate~$\Var[T]$ and~$\Var[\Tfp]$,
where~$T$ is the number of positive results
and~$\Tfp$ the number of false positive results.
Note that the number of false negatives~$\Tfn$ is zero in the noiseless case.

We are going to use the Efron--Stein inequality:
\begin{proposition}[Efron--Stein inequality]
  Let $X_1,\dotsc,X_n$ and $X'_1,\dotsc,X'_n$
  be independent variables, where~$X_i$
  has the same distribution as~$X'_i$ for all~$i$.
  Denote
  \begin{equation*}
    X
    :=(X_1,\dotsc,X_n)
    \qtextq{and}
    X^{(i)}
    :=(X_1,\dotsc,X_{i-1},X'_i,X_{i+1},\dotsc,X_n)\text.
  \end{equation*}
  Let~$f$ be a function of~$n$ variables.
  Then
  \begin{equation}
    \Var[f(X)]
    \le\frac12\sum_{i=1}^n\E\left[\abs[\big]{f(X)-f(X^{(i)})}^2\right]
    \text.
  \end{equation}
\end{proposition}

We can now prove Theorem~\ref{thm:variance}.
\begin{proof}[Proof of Theorem~\ref{thm:variance}]
  We pick an item~$i$,
  create an independent copy~$X'_i$ of its state~$X_i$,
  and write $X^{(i)}=(X_1,\dotsc,X_{i-1},X'_i,X_{i+1},\dotsc,X_n)$
  for the vector~$X$ where~$X_i$ is replaced by~$X'_i$.
  Since the multipool condition is symmetric under exchanging items,
  the Efron--Stein inequality and
  the tower property of the conditional expectation imply
  \begin{align}
    \Var[T]&
    \le\frac n2\E\left[\abs{T(X)-T(X_{(i)})}^2\right]\notag\\&
    \le\frac n2\E\left[\E\left[\abs{T(X)-T(X^{(i)})}^2
      \given\cF_i^c\right]\right]\notag\\&
    =n\rho(1-\rho)\E\left[\abs[\Big]{\ifu{\{\text{$Z_i=0$if$X_i=0$}\}}
      +\sum_{j\sim i}\ifu{\{\text{$j$ pivotal for $Z_i$}\}}}^2\right]
    \text.\label{eq:Efron-Stein_used}
  \end{align}
  Here, $\cF_i^c$ is the $\sigma$-algebra generated by
  $\{X_1,\dotsc,X_{i-1},X_{i+1},\dotsc,X_n\}$.
  The notation $j\sim i$ means that~$X_i$ and~$X_j$ share a pool,
  and~$i$ being pivotal for~$Z_j$ means that flipping~$X_i$
  from~$0$ to~$1$ will flip~$Z_j$ from~$0$ to~$1$.
  \par
  The item~$i$ shares pools with exactly $m(q-1)$ many other items,
  whence the sum has $m(q-1)+1$ terms.
  Using $\abs{\sum_{k = 1}^la_k}^2\le l\sum_{k=1}^l\abs{a_k}^2$
  in the expectation in~\eqref{eq:Efron-Stein_used}, we estimate further
  \begin{align}
    \Var[T]&
    \le n\rho(1-\rho)(m(q-1)+1)\left[\Prob\left(Z_i=0\given X_i=0\right)
      +\sum_{j\sim i}\Prob\left(\text{$i$ pivotal for $Z_j$}\right)\right]
    \notag\\&
    \le nmq\rho(1-\rho)
      \left[\spec+m(q-1)\Prob\left(\text{$i$ pivotal for $Z_j$}\right)\right]
    \text.\label{eq:Efron-Stein_used_2}
  \end{align}
  For~$i$ to be pivotal for~$Z_j$, we need~$X_j$ to be zero,
  the $q-2$ other items in the unique pool which contains items~$i$
  and~$j$ must be negative, and all other $m-1$ pools that item~$j$
  belongs to must have at least one positive item.
  This leads to
  \begin{equation}\label{eq:pivotal}
    \Prob\left(\text{$i$ pivotal for $Z_j$}\right)
    =(1-\rho)^{q-1}\left(1-(1-\rho)^{q-1}\right)^{m-1}\text.
  \end{equation}
  Inequality~\eqref{eq:Variance_T} now follows from~\eqref{eq:spec},
  \eqref{eq:Efron-Stein_used_2}, and \eqref{eq:pivotal}.
  \par
  The estimate~\eqref{eq:Variance_Tfp} on $\Var[\Tfp]$
  relies on analogous estimates using
  \begin{align*}&
    \Var[\Tfp]\\&
    \le n\rho(1-\rho)(m(q-1)+1)
      \left[\Prob\left(Z_i=1\given X_i=0\right)
      +\sum_{j\sim i}\Prob\left(\text{$i$ pivotal for $Z_j$}\right)\right]
      \notag\\&
    \le nmq\rho(1-\rho)\left[(1-\spec)
      +m(q-1)\Prob\left(\text{$i$ pivotal for $Z_j$}\right)\right]
    \text.\qedhere
  \end{align*}
\end{proof}

\section{Construction of maximally disjunct pooling matrices}\label{sec:matrices}

In this section, we prove Theorem~\ref{thm:matrices}
on the existence of particular multipools as in Definition~\ref{def:multipool},
that are maximally disjunct design matrices with constant row and column sums.
We are in the situation where~$q$ is a prime or a power of a prime and $n=q^2$
items are arranged in pools of size~$q$.
The underlying idea is essentially due to~\cite{ReedS-60}.
Multipools with a prime number of items per pool have been described,
\eg\ in \cite{Thierry-Mieg2006,Taeufer-20}
and with a prime power as number of items per pool in \cite{ErlichEtAl2015}.

\begin{proof}[Proof of Theorem~\ref{thm:matrices}]
  Let us first prove that the maximal number of pools an item participates in
  in a multipool arrangement with $n=q^2$ is $m=q+1$.
  By definition, two pools in a multipool can only share at most one item.
  Thus, any subset of two items in a pool uniquely determines this pool.
  Since there are at most~$\binom n2$ (unordered) pairs of items, we have at most
  \begin{equation*}
    \binom n2\binom q2^{-1}
    =\frac{n(n-1)}{q(q-1)}
  \end{equation*}
  many pools in a multipool.
  For the case $n=q^2$, this yields $q(q+1)=n+q$ many pools,
  whence every item can participate in at most $q+1$ many pools.
  This proves the bound on the maximal multiplicity~$m$
  in Theorem~\ref{thm:matrices}.

  For a construction of such a multipool with maximal multiplicity,
  let~$\F_q$ denote the finite field of size~$q$
  which exists because~$q$ is a prime power.
  We index the $n=q^2$
  items by the elements of the finite vector space~$\F_q^2$
  and use the ``straight lines'' in~$\F_q^2$ as pools:
  For all $a,b,c\in\F_q$, let
  \begin{equation*}
    \Pi_{a,b}:=\{(x,ax+b)\mid x\in\F_q\}
    \qtextq{and}
    \Pi_{\infty,c}:=\{(c,y)\mid y\in\F_q\}
    \text.
  \end{equation*}

  In any case, each line contains exactly~$q$ elements,
  each element of~$\F_q^2$ is contained in exactly $q+1$ lines,
  and different straight lines intersect in at most point of~$\F_q^2$.
  Counting the parameters $a,b,c$,
  one sees that there are $q^2+q$ straight lines in total.
  This yields a $(n,q,m)$-multipool of maximal multiplicity $m=q+1$.
  For multipools of lower multiplicity $m$,
  note that each $a\in\F_q$ determines a partition
  $\{\Pi_{a,b}\mid b\in\F_q\}$ of~$\F_q^2$.
  Therefore, for non-empty $M\subseteq\F_q$, the collection
  \begin{equation*}
    \Pi_M=\{\Pi_{a,b}\mid a\in M,b\in\F_q\}
  \end{equation*}
  is a multipool in which each item is contained
  in exactly~$m=\setsize M\in\{1,\dotsc,q\}$ pools.
\end{proof}

\begin{remark}
  If~$q$ is a proper prime, the arithmetic in~$\F_q$
  is simply the arithmetic modulo~$q$
  and these lines are indeed periodically continued straight lines
  with different slopes, cf.\ Figure~\ref{fig:pools-49-7-4}.
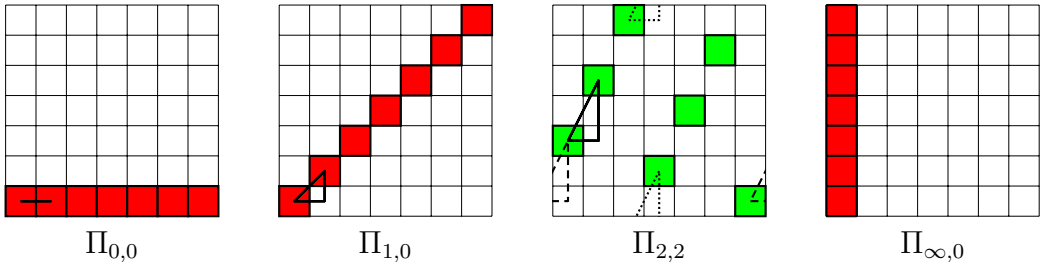
\begin{figure}[ht]
  \tikzsetnextfilename{pools-49-7-4}
  \begin{tikzpicture}[scale=.4]
    \foreach \m/\b [count=\n from 0] in {0/0, 1/0, 2/2}{
      \pgfmathtruncatemacro{\deltax}{\n*9}
      \begin{scope}[xshift=\deltax cm]
        \draw[very thin] (0,0) grid (7,7);
        \foreach \x in {0,...,6}{
          \pgfmathsetmacro\bcolor{array(\colorlist,\b)}
          \pgfmathtruncatemacro{\y}{Mod(\m*\x+\b,7)}
          \draw[thick,fill=\bcolor] (\x,\y) rectangle (\x+1,\y+1);
          \draw[thick,line join=round] 
            ([shift={(.5,.5)}] 0,\b) -| ([shift={(.5,.5)}] 1,\m+\b) -- cycle;
        }
        \draw (3.5,-1) node {$\Pi_{\m,\b}$};
      \end{scope}
    }
    \begin{scope}[xshift=18cm] 
      \clip (0,0) rectangle (7,7); 
      \begin{scope}[shift={(.5,.5)}]
        \draw[thick,line join=round,densely dotted] ( 2, 6) -| (3,8) -- cycle;
        \draw[thick,line join=round,densely dotted] ( 2,-1) -| (3,1) -- cycle;
        \draw[thick,line join=round,densely dashed] ( 6, 0) -| (7,2) -- cycle;
        \draw[thick,line join=round,densely dashed] (-1, 0) -| (0,2) -- cycle;
      \end{scope}
    \end{scope}
    \pgfmathsetmacro\bcolor{array(\colorlist,0)}
    \begin{scope}[xshift=27cm]
      \draw [very thin] (0,0) grid (7,7);
      \draw[fill=\bcolor] (0,0) rectangle (1,7);
      \draw[thick] (0,0) grid (1,7);
      \draw (3.5,-1) node {$\Pi_{\infty,0}$};
    \end{scope}
  \end{tikzpicture}
  \caption{The pools~$\Pi_{a,b}$ indeed are on
    periodically continued straight lines in the case $q=7$,
    since~$7$ is a proper prime.
    The gradient triangles illustrate the slope and the periodicity.}
  \label{fig:pools-49-7-4}
\end{figure}%
  The complete $(49,7,8)$-multipool is depicted
  in Figure~\ref{fig:multipool-49-7-8}.
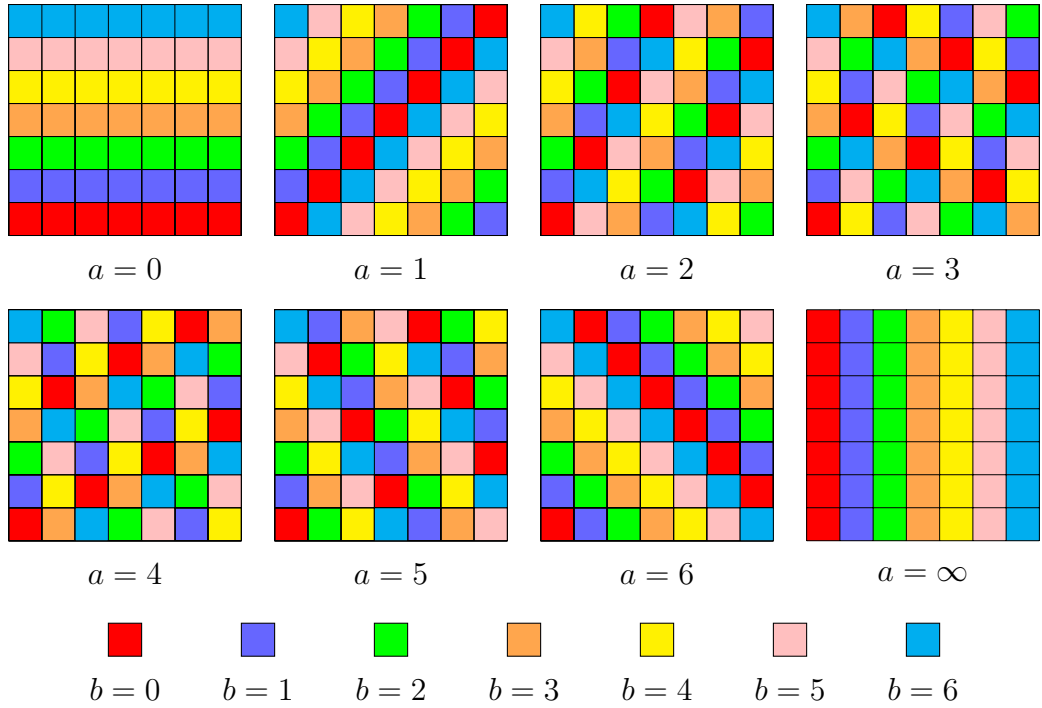
\begin{figure}[ht]
  \pgfmathtruncatemacro\p{7}
  \pgfmathtruncatemacro\dz{1}
  \pgfmathtruncatemacro\q{\p^\dz}
  \pgfmathtruncatemacro\n{\q^2}
  \pgfmathtruncatemacro\M{\q}
  \pgfmathtruncatemacro\withainfty{1}
  \pgfmathtruncatemacro\m{\M+\withainfty}
  \pgfmathtruncatemacro\t{\q*\m}
  \pgfmathsetmacro\s{3.5/(\q+1)}
  \tikzsetnextfilename{multipool-\n-\q-\m}
  \begin{tikzpicture}[scale=\s]
    \foreach \a [count=\nr from 0] in {0,...,\M-1}{
      \pgfmathsetmacro{\deltax}{Mod(\nr,4)*(\q+1)}
      \pgfmathsetmacro{\deltay}{div(\nr,4)*(\q+2.25)}
      \begin{scope}[xshift=\deltax cm, yshift=-\deltay cm]
        \foreach \b in {0,...,\q-1}{
          \pgfmathsetmacro\bcolor{array(\colorlist,\b)}
          \foreach \x in {0,...,\q-1}{
            \pgfmathtruncatemacro{\y}{polysum(Galoisprod(\a,\x,\p,\dz),\b,\p)}
            \fill[\bcolor] (\x,\y) rectangle (\x+1,\y+1);
          }
          \draw[very thin] (0,0) grid (\q,\q);
        }
      \draw (\q/2,-1) node {$a=\a$};
      \end{scope}
    }
    \ifthenelse{\equal\withainfty1}{
      \pgfmathsetmacro{\deltax}{Mod(\nr+1,4)*(\q+1)}
      \pgfmathsetmacro{\deltay}{div(\nr+1,4)*(\q+2.25)}
      \begin{scope}[xshift=\deltax cm, yshift=-\deltay cm]
        \foreach \c in {0,...,\q-1}{
          \pgfmathsetmacro\ccolor{array(\colorlist,\c)}
          \fill[\ccolor] (\c,0) rectangle (\c+1,\q);
        }
        \draw (\q/2,-1) node {$a=\infty$};
        \draw [very thin] (0,0) grid (\q,\q);
      \end{scope}
    }{}
    \pgfmathsetmacro{\deltay}{div(\nr+\withainfty,4)*(\q+2.25)+3.5}
    \begin{scope}[xshift=3cm,yshift=-\deltay cm]
      \foreach \b in {0,...,\q-1}{
        \pgfmathsetmacro\bcolor{array(\colorlist,\b)}
        \draw[fill=\bcolor] (4*\b,0) rectangle (4*\b+1,1);
        \node at (4*\b+.5,-1) {$b=\b$};
      }
    \end{scope}
  \end{tikzpicture}
  \caption{The $(\n,\q,\m)$-multipool.}
  \label{fig:multipool-\n-\q-\m}
\end{figure}%
  For prime powers, the more complicated arithmetic in~$\F_q$
  leads to a less tangible structure of these lines,
  \cf\ Figure~\ref{fig:pools-64-8-4}
  for an illustration in the case $q=8$
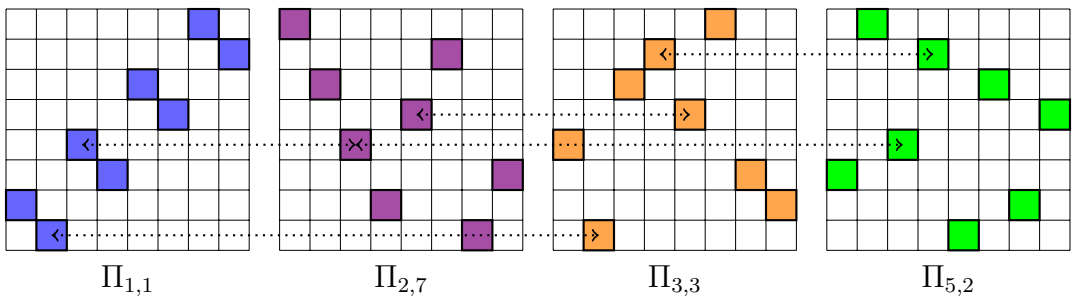
\begin{figure}[ht]
  \tikzsetnextfilename{pools-64-8-4}
  \begin{tikzpicture}[scale=.4]
    \foreach \m/\b [count=\n from 0] in {1/1, 2/7, 3/3, 5/2}{
      \pgfmathtruncatemacro{\deltax}{\n*9}
      \begin{scope}[xshift=\deltax cm]
        \draw[very thin] (0,0) grid (8,8);
        \foreach \x in {0,...,7}{
          \pgfmathsetmacro\bcolor{array(\colorlist,\b)}
          \pgfmathtruncatemacro{\y}{polysum(Galoisprod(\m,\x,2,3),\b,2)}
          \draw[thick,fill=\bcolor] (\x,\y) rectangle (\x+1,\y+1);
          \coordinate (\n\x) at (\x+.5,\y+.5);
        }
        \draw (4,-1) node {$\Pi_{\m,\b}$};
      \end{scope}
    }
    \draw[<->,thick,dotted] (02) -- (12);
    \draw[<->,thick,dotted] (12) -- (32);
    \draw[<->,thick,dotted] (01) -- (21);
    \draw[<->,thick,dotted] (14) -- (24);
    \draw[<->,thick,dotted] (23) -- (33);
  \end{tikzpicture}
  \caption{Multipools are no more on straight lines in the case $q=8$.
    The dotted arrows indicate the unique intersections between the pools.}
  \label{fig:pools-64-8-4}
\end{figure}%
  and Figure~\ref{fig:multipool-64-8-8}
  for an illustration of the $(64,8,8)$-multipool.
\begin{figure}[ht]
  \pgfmathtruncatemacro\p{2}
  \pgfmathtruncatemacro\dz{3}
  \pgfmathtruncatemacro\q{\p^\dz}
  \pgfmathtruncatemacro\n{\q^2}
  \pgfmathtruncatemacro\M{\q}
  \pgfmathtruncatemacro\withainfty{0}
  \pgfmathtruncatemacro\m{\M+\withainfty}
  \pgfmathtruncatemacro\t{\q*\m}
  \pgfmathsetmacro\s{3.5/(\q+1)}
  \tikzsetnextfilename{multipool-\n-\q-\m}
  \begin{tikzpicture}[scale=\s]
    \foreach \a [count=\nr from 0] in {0,...,\M-1}{
      \pgfmathsetmacro{\deltax}{Mod(\nr,4)*(\q+1)}
      \pgfmathsetmacro{\deltay}{div(\nr,4)*(\q+2.25)}
      \begin{scope}[xshift=\deltax cm, yshift=-\deltay cm]
        \foreach \b in {0,...,\q-1}{
          \pgfmathsetmacro\bcolor{array(\colorlist,\b)}
          \foreach \x in {0,...,\q-1}{
            \pgfmathtruncatemacro{\y}{polysum(Galoisprod(\a,\x,\p,\dz),\b,\p)}
            \fill[\bcolor] (\x,\y) rectangle (\x+1,\y+1);
          }
          \draw[very thin] (0,0) grid (\q,\q);
        }
      \draw (\q/2,-1) node {$a=\a$};
      \end{scope}
    }
    \ifthenelse{\equal\withainfty1}{
      \pgfmathsetmacro{\deltax}{Mod(\nr+1,4)*(\q+1)}
      \pgfmathsetmacro{\deltay}{div(\nr+1,4)*(\q+2.25)}
      \begin{scope}[xshift=\deltax cm, yshift=-\deltay cm]
        \foreach \c in {0,...,\q-1}{
          \pgfmathsetmacro\ccolor{array(\colorlist,\c)}
          \fill[\ccolor] (\c,0) rectangle (\c+1,\q);
        }
        \draw (\q/2,-1) node {$a=\infty$};
        \draw [very thin] (0,0) grid (\q,\q);
      \end{scope}
    }{}
    \pgfmathsetmacro{\deltay}{div(\nr+\withainfty,4)*(\q+2.25)+3.5}
    \begin{scope}[xshift=3cm,yshift=-\deltay cm]
      \foreach \b in {0,...,\q-1}{
        \pgfmathsetmacro\bcolor{array(\colorlist,\b)}
        \draw[fill=\bcolor] (4*\b,0) rectangle (4*\b+1,1);
        \node at (4*\b+.5,-1) {$b=\b$};
      }
    \end{scope}
  \end{tikzpicture}
  \caption{The $(\n,\q,\m)$-multipool.
    The vertical pools corresponding to $a=\infty$
    are missing to the maximal $(64,8,9)$-multipool.}
  \label{fig:multipool-\n-\q-\m}
\end{figure}
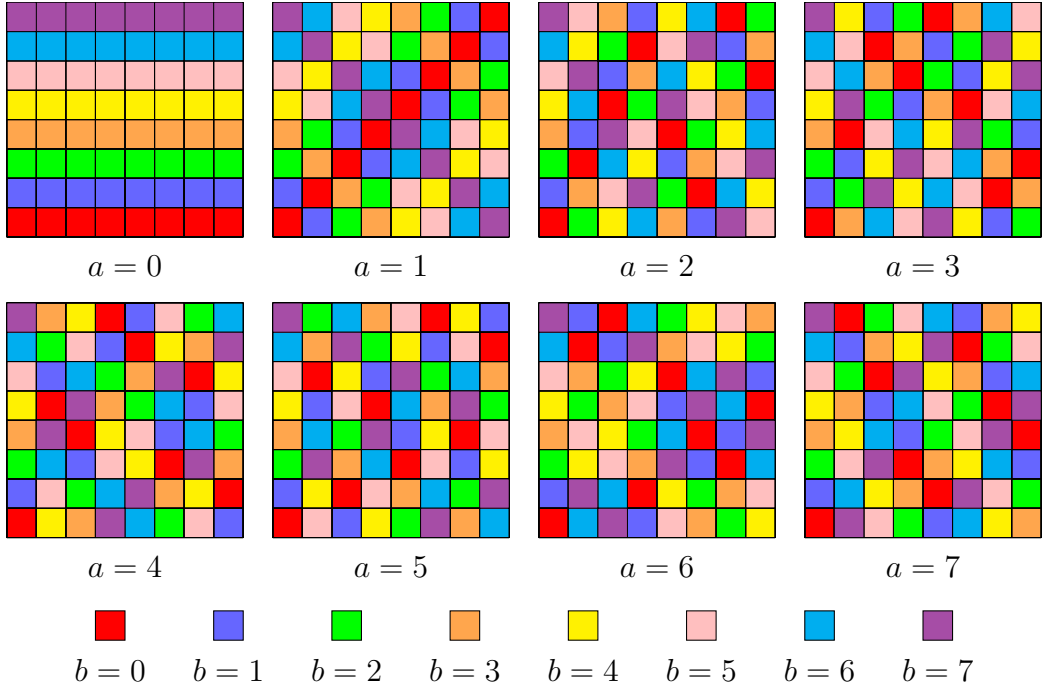%
%


Note that the design matrix
$A\in\{0,1\}^{(M\times\F_q)\times\F_q^2}=\{0,1\}^{t\times n}$
for the multipool~$\Pi_M$ is given by
\begin{align*}
  A_{(a,b),(x,y)}
  =\ifu{\Pi_{a,b}}((x,y))&
  =\begin{cases}
    1&(x,y)\in\Pi_{a,b}\\
    0&\text{otherwise}
  \end{cases}
\end{align*}
where $a,b\in\F_q$ determine the pool and $x,y\in\F_q$ index the item.
See Figure~\ref{fig:design7} for the design matrix of the $(49,7,8)$-multipool
\begin{figure}[ht]
  \pgfmathtruncatemacro\q{7}
  \pgfmathtruncatemacro\n{\q^2}
  \pgfmathtruncatemacro\M{\q}
  \pgfmathtruncatemacro\withainfty{1}
  \pgfmathtruncatemacro\m{\M+\withainfty}
  \pgfmathtruncatemacro\t{\q*\m}
  \pgfmathsetmacro\s{13/\n}
  \tikzsetnextfilename{design\q}
  \begin{tikzpicture}[scale=\s]
    \foreach \x in {0,...,\q-1}{
      \node[yshift=2ex] at ({\q*(\x+.5)},\t) {$x=\x$};
    }
    \foreach \a in {0,...,\M-1}{
      \node[xshift=-2ex,rotate=90] at (0,{\t-\q*(\a+.5)}) {$a=\a$};
      \foreach \b in {0,...,\q-1}{
        \foreach \x in {0,...,\q-1}{
          \pgfmathtruncatemacro\y{Mod(\a*\x+\b,\q)}
          \pgfmathtruncatemacro\X{\q*\x+\y}
          \pgfmathtruncatemacro\Y{\t-1-(\q*\a+\b)}
          \pgfmathsetmacro\bcolor{array(\colorlist,\b)}
          \draw[help lines, fill=\bcolor] (\X,\Y) rectangle (\X+1,\Y+1);
        }
      }
    }
    \ifthenelse{\equal\withainfty1}{
      \node[xshift=-2ex,rotate=90] at (0,{\t-\q*(\M+.5)}) {$a=\infty$};
      \foreach \c in {0,...,\q-1}{
        \pgfmathsetmacro\bcolor{array(\colorlist,\c)}
        \foreach \y in {0,...,\q-1}{
          \pgfmathtruncatemacro\X{\q*\c+\y}
          \pgfmathtruncatemacro\Y{\t-1-(\q*\M+\c)}
          \draw[help lines, fill=\bcolor] (\X,\Y) rectangle (\X+1,\Y+1);
        }
      }
    }{}
    \draw[help lines] (0,0) grid[step=\q] (\n,{\q*(\M+\withainfty)});
  \end{tikzpicture}
  \caption{Design matrix of the $(\n,\q,\m)$-multipool.}
  \label{fig:design\q}
\end{figure}
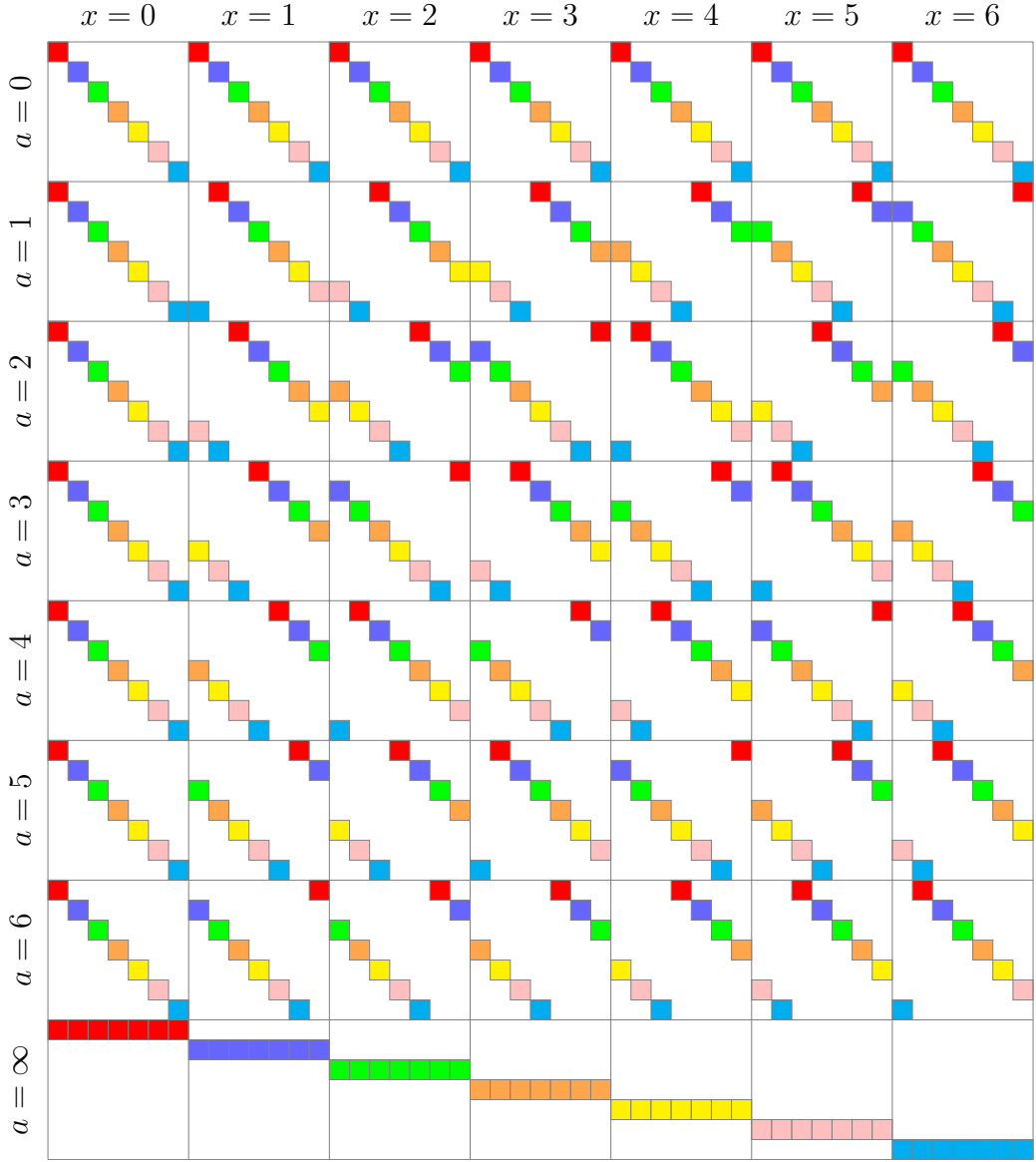%
and Figure~\ref{fig:design8} for the design matrix of the $(64,8,9)$-multipool.
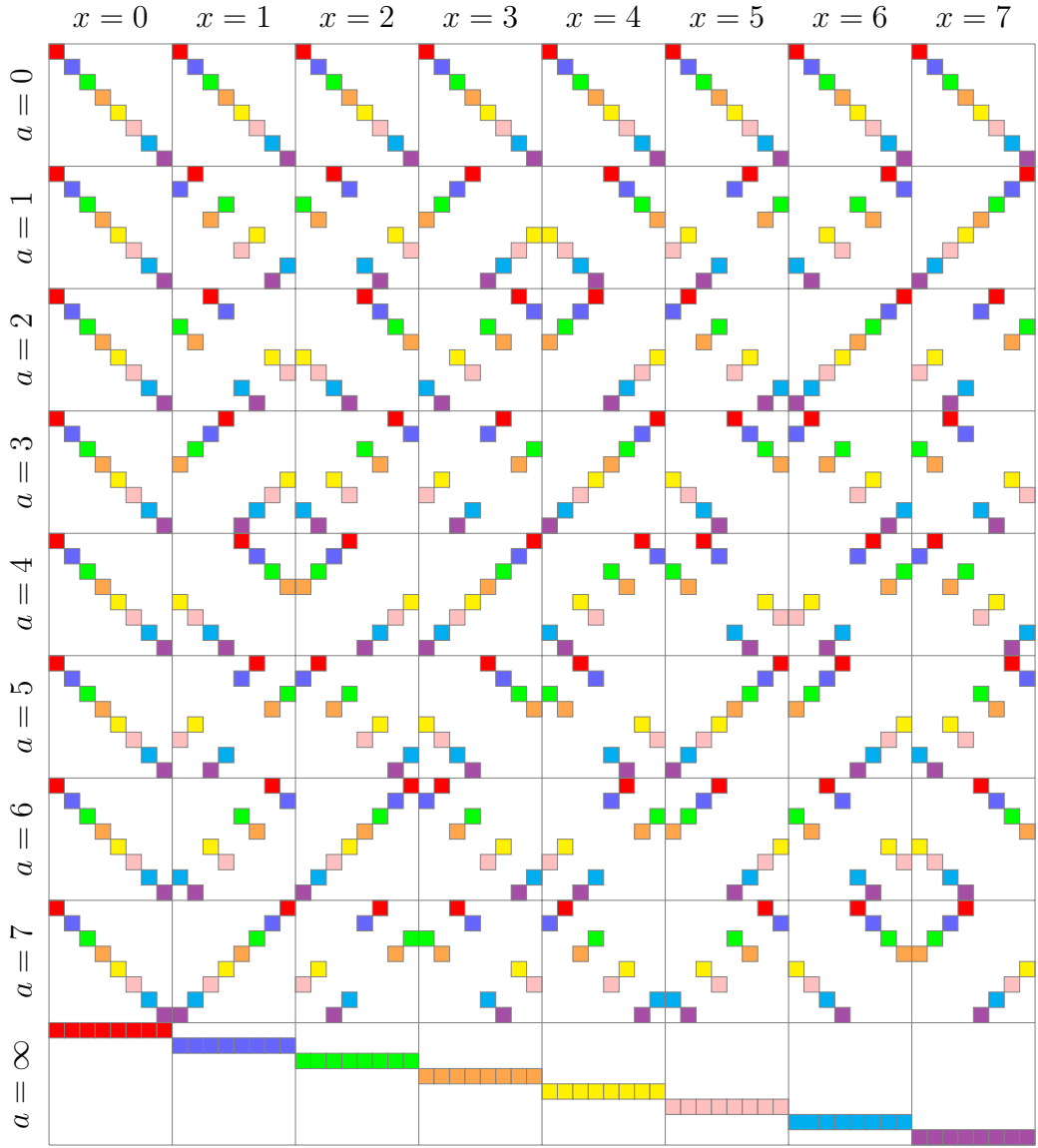
\begin{figure}[ht]
  \pgfmathtruncatemacro\p{2}
  \pgfmathtruncatemacro\dz{3}
  \pgfmathtruncatemacro\q{\p^\dz}
  \pgfmathtruncatemacro\n{\q^2}
  \pgfmathtruncatemacro\M{\q}
  \pgfmathtruncatemacro\withainfty{1}
  \pgfmathtruncatemacro\m{\M+\withainfty}
  \pgfmathtruncatemacro\t{\q*\m}
  \pgfmathsetmacro\s{13/\n}
  \tikzsetnextfilename{design\q}
  \begin{tikzpicture}[scale=\s]
    \foreach \x in {0,...,\q-1}{
      \node[yshift=2ex] at ({\q*(\x+.5)},\t) {$x=\x$};
    }
    \foreach \a in {0,...,\M-1}{
      \node[xshift=-2ex,rotate=90] at (0,{\t-\q*(\a+.5)}) {$a=\a$};
      \foreach \b in {0,...,\q-1}{
        \foreach \x in {0,...,\q-1}{
          \pgfmathtruncatemacro\y{polysum(Galoisprod(\a,\x,\p,\dz),\b,\p)}
          \pgfmathtruncatemacro\X{\q*\x+\y}
          \pgfmathtruncatemacro\Y{\t-1-(\q*\a+\b)}
          \pgfmathsetmacro\bcolor{array(\colorlist,\b)}
          \draw[help lines, fill=\bcolor] (\X,\Y) rectangle (\X+1,\Y+1);
        }
      }
    }
    \ifthenelse{\equal\withainfty1}{
      \node[xshift=-2ex,rotate=90] at (0,{\t-\q*(\M+.5)}) {$a=\infty$};
      \foreach \c in {0,...,\q-1}{
        \pgfmathsetmacro\bcolor{array(\colorlist,\c)}
        \foreach \y in {0,...,\q-1}{
          \pgfmathtruncatemacro\X{\q*\c+\y}
          \pgfmathtruncatemacro\Y{\t-1-(\q*\M+\c)}
          \draw[help lines, fill=\bcolor] (\X,\Y) rectangle (\X+1,\Y+1);
        }
      }
    }{}
    \draw[help lines] (0,0) grid[step=\q] (\n,{\q*(\M+\withainfty)});
  \end{tikzpicture}
  \caption{Design matrix of the $(\n,\q,\m)$-multipool.}
  \label{fig:design\q}
\end{figure}%
\end{remark}

\begin{remark}\label{rem:algebraic}
  An equivalent, more algebraic construction of pooled tests
  based on Reed--Solomon codes is described in \cite{ErlichEtAl2015}.
  There, the items are indexed by polynomials with coefficients in~$\F_q$,
  more specifically,
  the~$n$ smallest polynomials according to the lexicographical order.
  To understand the base ordering in~$\F_q$ giving rise to this order,
  we need more details of the standard construction of~$\F_q$.

  Consider the factoring of the prime power $q=p^a$.
  For this, we fix an irreducible polynomial~$P_{p,a}$ of degree~$a$ over~$\F_p$,
  for example the Conway polynomial that fits the bill,
  see e.\,g.\ \cite{Luebeck2008}.
  Then, $\F_q:=\F_p[x]/P_{p,a}$ is the space of all polynomials over~$\F_p$
  modulo~$P_{p,a}$.
  Consider the representative of an element $P\in\F_q$
  with degree lower than~$a$.
  The coefficients can be used to form a number in base~$p$
  which we can computed by considering~$P$
  as a polynomial over~$\Z$ and evaluating it at~$p$.
  This defines a bijection
  \begin{equation}
    \F_q\to\{0,1,\dotsc,q-1\}\text,
  \end{equation}
  which gives the ordering of~$\F_q$.

  The pools are grouped into~$m$ layers, each layer consisting of~$q$ pools.
  For convenience, we index the pools again by $M'\times\F_q$,
  where~$M'\subseteq\F_q$ has size~$m$.
  The design matrix~$A'\in\{0,1\}^{(M'\times\F_q)\times\F_q^2}$
  of the Reed--Solomon code testing strategy is given by
  \begin{equation*}
    A'_{(a,b),f}
    =\begin{cases}
      1&f(a)=b\\
      0&\text{otherwise.}
    \end{cases}
  \end{equation*}
  If~$n$ is a multiple of~$q$ and lower or equal to~$q^2$,
  this construction gives a multipool, too.
  Indeed, each item is in exactly one pool for each layer,
  so it is contained in exactly~$m$ pools.
  Furthermore, by the lexicographical order,
  the absolute terms of the polynomials cycle through~$\F_q$
  exactly~$n/q$ times, so that each pool contains exactly $n/q$ items.
  And finally, because $n\le q^2$, all polynomials used to index the items
  are constant or linear, whence two items can share at most one pool.
  Indeed, if two polynomials of degree lower or equal to~$1$
  agree in two or more places, they are equal and describe the same pool.

  This gives us multipools with multiplicities up to $m = q$.
  The remaining layer in a multipool of maximal multiplicity $m=q+1$
  consists of lines of slope infinity and cannot be represented
  as a graph of polynomials, so they need to be added manually.

  The geometric construction and the algebraic construction from
  Remark~\ref{rem:algebraic} produce the same multipools.
  Indeed, the geometric condition $y=ax+b$
  corresponds to the algebraic condition $y-ax=b$, thus choosing $M'=-M$,
  the design matrices agree up to the order of rows.
\end{remark}

\printbibliography

\end{document}